\documentclass[11pt,onecolumn,a4paper]{article}

 \usepackage{amsfonts}
 \usepackage{amsmath}
 \usepackage{amssymb}
 \usepackage{graphicx}
 \usepackage{color}
 \usepackage[english]{babel}
 \usepackage{amsmath,amssymb,booktabs,ctable,graphicx,epsf,epsfig }
\usepackage{psfrag,pstricks}
\usepackage{setspace}

\renewcommand{\c}{\textcolor{black}}

\usepackage[utf8]{inputenc}

\DeclareFontFamily{U}{mathb}{\hyphenchar\font45}
\DeclareFontShape{U}{mathb}{m}{n}{
<-6> mathb5 <6-7> mathb6 <7-8> mathb7
<8-9> mathb8 <9-10> mathb9
<10-12> mathb10 <12-> mathb12
}{}
\DeclareSymbolFont{mathb}{U}{mathb}{m}{n}
\DeclareMathSymbol{\llcurly}{\mathrel}{mathb}{"CE}
\DeclareMathSymbol{\ggcurly}{\mathrel}{mathb}{"CF}

 \headsep
 \baselineskip
\date{}

\begin{document}

\newtheorem{definition}{Definition}[section]
\newtheorem{prop}{Proposition}[section]
\newtheorem{thm}[prop]{Theorem}
\newtheorem{lem}[prop]{Lemma}
\newtheorem{cor}[prop]{Corollary}
\newtheorem{rem}{Remark}[section]
\newtheorem{ex}{Example}[section]
\newtheorem{proof}[prop]{Proof}
\renewcommand{\c}{\textcolor{blue}}

\title{\textbf{An Addendum to the Problem of Zero-Sum LQ Stochastic Mean-Field Dynamic Games\\ (Extended version)}}

\author{Samir Aberkane$^{1,2}$ and Vasile Dragan$^{3,4}$
\\ [1.mm]
\small\it $^1$ Université de Lorraine, CRAN, UMR 7039, Campus Sciences, BP 70239\\[-1mm]
\small\it VandÏuvre-les-Nancy Cedex, 54506, France\\[-1mm]
\small\it $^2$ CNRS, CRAN, UMR 7039, France\\[-1mm]
\small {\tt samir.aberkane@univ-lorraine.fr}\\
\small\it $^3$ Institute of Mathematics "Simion Stoilow"\\[-1mm]
\small\it of the Romanian Academy \\[-1.mm]
\small\it P.O.Box 1-764, RO-014700, Bucharest, Romania   \\ [-1.mm]
\small {\tt Vasile.Dragan@imar.ro}\\
\small\it $^4$ The Academy of the Romanian Scientists, Str. Ilfov, 3, Bucharest, Romania}

\date{}
\maketitle

\begin{abstract}
In this paper, we first address a linear quadratic mean-field game problem with a leader-follower structure. By adopting a Riccati-type approach, we show how one can obtain a state-feedback representation of the pairs of strategies which achieve an open-loop Stackelberg equilibrium in terms of the global solutions of a system of coupled matrix differential Riccati-type equations.
In the second part of this paper, we obtain necessary and sufficient conditions for the solvability of the involved coupled generalized Riccati equations.
\end{abstract}


\section{Introduction}
Mean-field stochastic differential equations (SDEs), also called McKean--Vlasov SDEs, can be used to effectively characterize dynamical systems of large populations subject to a mean-field interaction (one can refer to \cite{Kac} and \cite{McKean} for pioneering works about the subject). Recently, the mean-field/McKean--Vlasov SDEs has been wildly used in mean-field game theory. This theory attracted a huge interest from the scientific community these last few years since the pioneering works \cite{Huang:06,Lasry:06,Lasry:06-1,Lasry:06-2}. One can refer for example to the textbooks \cite{Bensoussan,Carmona,Caines,Gueant} and the references therein. A particularly attractive subclass of mean-field/McKean--Vlasov type differential games is stochastic linear quadratic (LQ) mean field games. This is due to its tractable analytical structure \cite{Barreiro,Bensoussan1,Graber,Huang2,Li,Moon:20,Sun:21bis,Tian}.\\
In this paper we will focus particularly on LQ zero-sum differential games driven by SDEs with McKean--Vlasov type. More specifically, we are studying in the present work the problem of the existence of an open-loop Stackelberg equilibrium for a two persons zero-sum LQ McKean--Vlasov type differential game. It is known that the concept of Stackelberg equilibrium is defined for a differential game with a hierarchical structure in decision making, often named \textit{leader-follower structure} of the considered game. We show how one can obtain a state-feedback representation of the pairs of strategies which achieve an open-loop Stackelberg equilibrium in terms of the global solutions of a system of coupled generalized coupled matrix differential Riccati-type equations with specific sign constraints on the quadratic parts of such nonlinear equations. Such constraints lead to a sign indefiniteness of the quadratic parts of the involved Riccati equations. We mention here that such a system of coupled Riccati equations has been first introduced by \cite{Yong:13} in order to solve an LQ mean-field control problem. The solvability of these two Riccati equations is established in \cite{Yong:13} under certain positivity conditions and \cite{Sun:17} showed that the solvability is equivalent to the uniform convexity of the cost functional. Note however that in the case of zero-sum LQ McKean--Vlasov type differential games, there are few results in the literature regarding the solvability of the associated Riccati equations with indefinite sign of there quadratic parts. We cite here the efforts made by \cite{Moon:20,Tian} in order to tackle this problem. A review of the solvability conditions proposed in \cite{Moon:20,Tian} showed that they are obtained under constraints on the system's parameters. Also, the proposed existence conditions are only sufficients. Most importantly, the given conditions do not take into account the sign conditions constraints imposed on the quadratic part of the Riccati equations involved in the solution of the LQ dynamic game. Knowing that these sign conditions play a key role in the solution process, this represents in our opinion a strong limitation for such results. Very recently, \cite{Sun:21bis} addressed such a challenging problem under the general framework of \textit{uniform convexity/concavity}. The solvability conditions proposed by  \cite{Sun:21bis} take explicitly into account the sign constraints on the quadratic terms of the Riccati equations. However the author succeeded only to give \textit{sufficient solvability conditions}. One of the main contributions of this paper is to propose \textit{necessary and sufficient solvability conditions}. We also adopt in this paper the uniform convexity/concavity framework making the obtained results as general as the ones obtained in \cite{Sun:21bis}. By adequately defining some auxiliary performances criteria, we show that the uniform convexity/concavity of such functionals are equivalent with the solvability of the considered system of matrix Riccati-differential equations.\\
To summarize, we list below the main contributions of the paper:
\begin{itemize}
\item [i)] We have addressed the problem of the existence of an open-loop Stackelberg equilibrium for a two persons zero-sum LQ McKean--Vlasov type differential game. We show how one can obtain a state-feedback representation of the pairs of strategies which achieve an open-loop Stackelberg equilibrium in terms of the global solutions of a system of coupled generalized coupled matrix differential Riccati-type equations (see \textbf{Theorem 2.2}).
\item [ii)] In the main result of Section 3 (\textbf{Theorem 3.15}), we provided necessary and sufficient conditions for the global existence of the solution of the considered system of coupled Riccati equations verifying specific sign conditions. These conditions are formulated under the general uniform convexity/concavity framework. 
\item [iii)] We have shown that the class of Riccati equations considered in \cite{Sun:21bis} can be viewed as a particular case of the problem addressed in this paper. As a matter of fact, by specializing the result given in Theorem 3.15 we have obtained necessary and sufficient conditions counterpart (see \textbf{Theorem 3.17}) of the results given by Theorem 4.2 and 4.4 from \cite{Sun:21bis} where only sufficient conditions have been proposed.
\end{itemize}        
\noindent This paper is organized as follows: In Section 2 we formulate and solve the game problem. In Section 3 the existence conditions of the solution of the involved Riccati equations are given. Section 4 concludes this paper.\\
\noindent\textbf{Notations}. $A^\top$ stands for the transpose of the matrix $A$. In block matrices, $\star$ indicates symmetric terms: $\left(%
\begin{array}{cc}
 A & B \\
  B^T & C \\
\end{array}%
\right)=\left(%
\begin{array}{cc}
  A & \star \\
  B^T & C \\
\end{array}%
\right)=\left(%
\begin{array}{cc}
  A & B \\
  \star & C \\
\end{array}%
\right)$.
The expression $MN\star$ is equivalent to $MNM^T$ while $M\star$ is equivalent to $MM^T$. 
\section{The game}
Consider the following controlled linear mean-field type stochastic differential equation (SDE):
\begin{equation}\label{e.1}
\begin{cases}
dx(t)=\left\{A_0(t)x(t)+\bar{A}_0(t)\mathbb{E}[x(t)]+B_0(t)u(t)+\bar{B}_0(t)\mathbb{E}[u(t)]\right\}dt\\
\quad+\sum_{j=1}^r\left\{A_j(t)x(t)+\bar{A}_j(t)\mathbb{E}[x(t)]+B_j(t)u(t)+\bar{B}_j(t)\mathbb{E}[u(t)]\right\}dw_j(t)\\
x(s)=x_s\in \mathbb{R}^n
\end{cases}
\end{equation}
$t\in[s,\; T]$.
In ({\ref{e.1}}), $\{w(t)\}_{t\geq 0}\;(w(t)=(w_1(t),...,w_r(t))^{\top})$ is an $r$-dimensional standard Wiener process defined on a given probability space $(\Omega,{\cal{F}},{\cal{P}})$.
Let $\mathcal{F}_t$, $t \geq 0$ denotes the family of $\sigma$-algebras $\mathcal{F}_t =\sigma(w(s),0 \leq s \leq t)$. We denote by $\mathbb{E} [\cdot] $ the mathematical expectation.
We assume that: $t \rightarrow A_k(t):[0,\; T] \rightarrow{\mathbb{R}}^{n\times n}$, $t \rightarrow \bar{A}_k(t):[0,\; T] \rightarrow{\mathbb{R}}^{n\times n}$, $t\rightarrow B_k(t):[0,\; T]\rightarrow{\mathbb{R}}^{n\times m}$, $t\rightarrow \bar{B}_k(t):[0,\; T]\rightarrow{\mathbb{R}}^{n\times m}$, $0\leq k\leq r$, are continuous matrix valued  functions.
We define the class ${\cal{U}}_{ad}$ of the admissible controls consisting of all stochastic processes, $u=\{u(t)\}_{t\geq s}\in L_{w}^2([s,T];{\mathbb{R}}^m)$.
In this work,  the space $L_{w}^2({\cal{I}};{\mathbb{R}}^d)$ stands  for the vector space  of stochastic processes $v:{\cal I}\times \Omega\to {\mathbb R}^d$ which are non-anticipative with respect to the filtration $\{{\cal F}_t\}_{t\geq 0}$ and satisfy
$$\int_{\cal I} {\mathbb E}[|v(t)|^2]dt < \infty.$$ For more details regarding the properties of such  stochastic processes, we refer to section 1.9 from \cite{carte2013}.
Invoking  Proposition 2.1 in \cite{yong-2017}, one deduces that  ({\ref{e.1}}) admits a unique solution $x(\cdot)\in L_{w}^2([s,T];{\mathbb{R}}^n)$.\\
\noindent To (\ref{e.1}) we associate the following quadratic cost functional on finite horizon time:
\begin{align}\label{e.3}
\mathcal{J}(s,T,x_s;u)&=\frac{1}{2}\mathbb{E}\left\{\langle G_Tx(T), x(T)\rangle+\langle\bar{G}_T\mathbb{E}[x(T)],\;\mathbb{E}[x(T)]\rangle\right.\notag\\
&\left.+\int_s^T\left(\begin{array}{c}x_u(t) \\u(t) \\ \mathbb{E}[x_u(t)]\\\mathbb{E}[u(t)]\end{array}\right)^{\top}\left(\begin{array}{cccc}M(t) & L(t) & \mathbf{0} & \mathbf{0} \\L^{\top}(t) & R(t) & \mathbf{0} & \mathbf{0} \\\mathbf{0} & \mathbf{0} & \bar{M}(t) & \bar{L}(t) \\\mathbf{0} & \mathbf{0} & \bar{L}^{\top}(t) & \bar{R}(t)\end{array}\right)\star dt\right\}
\end{align}
where $x_u(t)$, $t\geq s\geq 0$ is the solution of the initial problem (\ref{e.1}) determined by the input $u(t)$, $G_T$, $\bar{G}_T$ $\in \mathcal{S}_n$ and $t\rightarrow L(t):[0,\; T]\rightarrow{\mathbb{R}}^{n\times m}$, $t\rightarrow \bar{L}(t):[0,\; T]\rightarrow{\mathbb{R}}^{n\times m}$, $t\rightarrow M(t):[0,\; T]\rightarrow{\cal{S}}_n$, $t\rightarrow \bar{M}(t):[0,\; T]\rightarrow{\cal{S}}_n$, $t\rightarrow R(t):[0,\; T]\rightarrow{\cal{S}}_m$, $t\rightarrow \bar{R}(t):[0,\; T]\rightarrow{\cal{S}}_m$ are continuous matrix valued  functions. Here and in the sequel, ${\cal S}_q\subset {\mathbb R}^{q\times q}$ stands for the vector space of symmetric matrices of dimension $q\times q$. Obviously, for any $u=\{u(t)\}_{t\geq s}\in L_{w}^2([s,T];{\mathbb{R}}^m)$, $\mathcal{J}(s,T,x_s;u)$ is well defined.\\
For any $t\in [0,\; T]$ we set:
\begin{align}\label{e.4}
&B_k(t)=\left(
                   \begin{array}{cc}
                     B_{k1}(t) & B_{k2}(t) \\
                   \end{array}
                 \right),
B_{kj}(t)\in{\mathbb{R}}^{n\times m_j},\notag\\
&\bar{B}_k(t)=\left(
                   \begin{array}{cc}
                     \bar{B}_{k1}(t) & \bar{B}_{k2}(t) \\
                   \end{array}
                 \right),
\bar{B}_{kj}(t)\in{\mathbb{R}}^{n\times m_j}, 0\leq k\leq r,
L(t)=\left( L_1(t)\quad L_2(t)\right), \nonumber\\ &L_j(t)\in{\mathbb R}^{n\times m_j}, \bar{L}(t)=\left( \bar{L}_1(t)\quad \bar{L}_2(t)\right), \bar{L}_j(t)\in{\mathbb R}^{n\times m_j} , j=1,2,\notag\\
&R(t)=\left(
                                                                                \begin{array}{cc}
                                                                           R_{11}(t) & R_{12}(t) \\
                                                                                  R_{12}^{\top}(t) & R_{22}(t) \\
                                                                                \end{array}
                                                                              \right),
R_{lj}(t)\in{\mathbb{R}}^{m_l\times m_j}, \bar{R}(t)=\left(
                                                                                \begin{array}{cc}
                                                                                  \bar{R}_{11}(t) & \bar{R}_{12}(t) \\
                                                                                  \bar{R}_{12}^{\top}(t) & \bar{R}_{22}(t) \\
                                                                                \end{array}
                                                                              \right),\notag\\
                                                                              &\bar{R}_{lj}(t)\in{\mathbb{R}}^{m_l\times m_j},  l,j=1,2.
\end{align}
In order to ease the description of the game, we rewrite (\ref{e.1}) and (\ref{e.3}) according to the partition $u=\left(\begin{array}{cc}u_1^{\top} & u_2^{\top}\end{array}\right)^{\top}$ of the input $u$ and the partition above of the coefficients. Thus we obtain:
\begin{equation}\label{e.5}
\begin{cases}
dx(t)=\left\{A_0(t)x(t)+\bar{A}_0(t)\mathbb{E}[x(t)]+B_{01}(t)u_1(t)+\bar{B}_{01}(t)\mathbb{E}[u_1(t)]\right.\\
\quad\left.+B_{02}(t)u_2(t)+\bar{B}_{02}(t)\mathbb{E}[u_2(t)]\right\}dt\\
\quad+\sum_{j=1}^r\left\{A_j(t)x(t)+\bar{A}_j(t)\mathbb{E}[x(t)]+B_{j1}(t)u_1(t)+\bar{B}_{j1}(t)\mathbb{E}[u_1(t)]\right.\\
\quad\left.+B_{j2}(t)u_2(t)+\bar{B}_{j2}(t)\mathbb{E}[u_2(t)]\right\}dw_j(t)\\
x(s)=x_s\in \mathbb{R}^n
\end{cases}
\end{equation}
\begin{align}\label{e.6}
\mathcal{J}(s,T,x_s;u_1(\cdot),u_2(\cdot))&=\frac{1}{2}\mathbb{E}\left\{\langle G_Tx(T), x(T)\rangle+\langle\bar{G}_T\mathbb{E}[x(T)],\;\mathbb{E}[x(T)]\rangle\right.\notag\\
&\left.+\int_s^T\left\{\left(\begin{array}{c}x_u(t) \\u_1(t) \\u_2(t)\end{array}\right)^{\top}\left(\begin{array}{ccc}M(t) & L_1(t) & L_2(t) \\\star & R_{11}(t) & R_{12}(t) \\\star & \star & R_{22}(t)\end{array}\right)\star\right.\right.\notag\\
&\left.\left.+\left(\begin{array}{c}\mathbb{E}[x_u(t)] \\\mathbb{E}[u_1(t)] \\\mathbb{E}[u_2(t)]\end{array}\right)^{\top}\left(\begin{array}{ccc}\bar{M}(t) & \bar{L}_1(t) & \bar{L}_2(t) \\\star & \bar{R}_{11}(t) & \bar{R}_{12}(t) \\\star & \star & \bar{R}_{22}(t)\end{array}\right)\star\right\} dt\right\}
\end{align}
As usual, the inputs $u_k:[s,\; T]\rightarrow \mathbb{R}^{m_k}$ will be called strategies available for player $k$, $k=1,\;2$. These are functions with the property that $\left(\begin{array}{cc}u_1^{\top}(\cdot) & u_2^{\top}(\cdot)\end{array}\right)^{\top}$ are in a subset of $\mathcal{U}_{\text{ad}}$. Such a subset of $\mathcal{U}_{\text{ad}}$ will be called the set of admissible strategies. In the description of a set of admissible strategies one takes into account the type of information available for player $k$ in order to compute the value of $u_k(t)$.\\
\noindent In a two player zero-sum differential game, the two players have opposite aims. In the present work, we assume that the aim of the Player $1$ is to choose the strategies $u_1(\cdot)$ from the set of its admissible strategies in order to maximize the value of the objective function, while the Player $2$ chooses its strategies $u_2(\cdot)$ from the set of its admissible strategies in order to minimize the value of the objective function. Roughly speaking, the solution of a problem described by a zero-sum differential game with two players is a pair of strategies $(\tilde{u}_1(\cdot),\;\tilde{u}_2(\cdot))$ named a pair of equilibrium strategies. Such a pair of strategies, if it exists, is the best choice for both players according to the adopted definition. Among the most frequently studied types of equilibria are Nash equilibrium strategy and Stackelberg equilibrium strategy.\\
In the present work we are studying the problem of the existence of an open-loop Stackelberg equilibrium for a two persons zero-sum linear quadratic mean-field differential game described by the controlled system (\ref{e.5}) and the objective function (\ref{e.6}). It is known that the concept of Stackelberg equilibrium is defined for a differential game with a hierarchical structure in decision making, often named \textit{leader-follower structure} of the considered game.\\
In this work, we consider the case when Player 1 is the leader and the Player 2 is the follower. this means that Player 1 is the first to announce its strategy $u_1(\cdot)\in \mathcal{U}_1\triangleq L_w^2\left\{[s,T];\; \mathbb{R}^{m_1}\right\}$. Next, Player 2 will choose a strategy $\tilde{u}_2(\cdot ;u_1(\cdot))$ (depending upon the announced strategy $u_1(\cdot)$ of the leader) which minimizes the mapping $u_2(\cdot)\rightarrow \mathcal{J}(s,T,x_s;u_1(\cdot),u_2(\cdot)):\; L_w^2\left\{[s,T];\; \mathbb{R}^{m_2}\right\}\rightarrow \mathbb{R}$.\\
Knowing the strategy adopted by the follower, the leader will choose a strategy $\tilde{u}_1(\cdot)$ in order to maximize the value of the mapping $u_1(\cdot)\rightarrow \mathcal{J}(s,T,x_s;u_1(\cdot),\tilde{u}_2(\cdot;u_1(\cdot))):\; L_w^2\left\{[s,T];\; \mathbb{R}^{m_1}\right\}\rightarrow \mathbb{R}$.\\
Setting $\tilde{u}_2(\cdot)\triangleq \tilde{u}_2(\cdot;\tilde{u}_1(\cdot))$ one obtains a pair $\left(\tilde{u}_1(\cdot),\tilde{u}_2(\cdot)\right)$ which forms an open-loop Stackelberg equilibrium for the considered zero-sum differential game.\\
The procedure briefly described above suggests the following definition for an open-loop Stackelberg equilibrium:
\begin{definition}\label{D2.1}
A pair of admissible strategies $\left(\tilde{u}_1(\cdot),\tilde{u}_2(\cdot)\right)$ achieves an open-loop Stackelberg equilibrium for the two players zero-sum mean-field differential game described by the controlled system (\ref{e.5}), the performance criterion (\ref{e.6})
and the set of admissible strategies $\mathcal{U}_{\text{ad}}=\mathcal{U}_1\times \mathcal{U}_2$ where $\mathcal{U}_k=L_w^2\left\{[s,T];\; \mathbb{R}^{m_k}\right\}, k=1,2$, if
\begin{equation}\label{e.7}
\underset{u_1(\cdot)\in \mathcal{U}_1}{\sup}\underset{u_2(\cdot)\in \mathcal{U}_2}{\inf}\mathcal{J}(s,T,x_s;u_1(\cdot),u_2(\cdot))=\mathcal{J}(s,T,x_s;\tilde{u}_1(\cdot),\tilde{u}_2(\cdot))=\underset{u_2(\cdot)\in \mathcal{U}_2}{\inf}\mathcal{J}(s,T,x_s;\tilde{u}_1(\cdot),u_2(\cdot))
\end{equation}
\end{definition}
\noindent Since the pioneering Stackelberg's work \cite{Stackelberg:34} the theory of Stackelberg games was developed in connection with applications in economy, finance, engineering and so on.
A historical perspective on this topic my be found in \cite{Sun:21}. In the present work, our aim  is to show how one can obtain a state-feedback representation of the pairs of strategies which achieve an open-loop Stackelberg equilibrium.
To this end, we consider the following system of terminal value problems (TVPs) associated to coupled matrix differential Riccati-type equations:
\begin{equation}\label{e.9}
\begin{cases}
-{\dot{X}}(t)\!=\!A_0^{\top}(t)X(t)\!+\!X(t)A_0(t)\!+\!\sum\limits_{k=1}^rA_k^{\top}(t)X(t)A_k(t)\!+M(t)\\
-[X(t)B_0(t)\!+\!\!\sum\limits_{k=1}^r\!A_k^{\top}(t)X(t)B_k(t)+L(t)][R(t)\\
+\sum\limits_{k=1}^rB_k^{\top}(t)X(t)B_k(t)]^{-1}[B_0^{\top}(t)X(t)\!+\!\!\sum\limits_{k=1}^rB_k^{\top}(t)X(t)A_k(t)+L^{\top}(t)]\\
X(T)=G_T
\end{cases}
\end{equation}
and
\begin{equation}\label{e.10}
\begin{cases}
-{\dot{\hat{X}}}(t)\!=\!\hat{A}_0^{\top}(t)\hat{X}(t)\!+\!\hat{X}(t)\hat{A}_0(t)\!+\!\sum\limits_{k=1}^r\hat{A}_k^{\top}(t)X(t)\hat{A}_k(t)\!+\hat{M}(t)\\
-[\hat{X}(t)\hat{B}_0(t)\!+\!\!\sum\limits_{k=1}^r\!\hat{A}_k^{\top}(t)X(t)\hat{B}_k(t)+\hat{L}(t)][\hat{R}(t)\\
+\sum\limits_{k=1}^r\hat{B}_k^{\top}(t)X(t)\hat{B}_k(t)]^{-1}[\hat{B}_0^{\top}(t)\hat{X}(t)\!+\!\!\sum\limits_{k=1}^r\hat{B}_k^{\top}(t)X(t)\hat{A}_k(t)
+\hat{L}^{\top}(t)]\\
\hat{X}(T)=\hat{G}_T
\end{cases}
\end{equation}
$t\in[0,\; T]$ with the unknown functions $t\rightarrow {X}(t)$, $X(t)=X^{\top}(t)$, and $t\rightarrow {\hat{X}}(t)$, $\hat{X}(t)=\hat{X}^{\top}(t)$, where $\hat{A}_k(\cdot)=A_k(\cdot)+\bar{A}_k(\cdot)$ and $\hat{B}_k(\cdot)=B_k(\cdot)+\bar{B}_k(\cdot)$, $0\leq k\leq r$, $\hat L(\cdot)=L(\cdot)+\bar L(\cdot), \hat R(\cdot)= R(\cdot)+\bar R(\cdot), \hat M(\cdot)= M(\cdot)+\bar M(\cdot)$.\\
Let $\mathcal{I}_X\subset [0,\; T]$ ($\mathcal{I}_{\hat X}\subset [0,\; T]$) be the maximal interval where the solution to the TVP (\ref{e.9}) (the solution to the TVP (\ref{e.10})) is defined.\\
We denote:
\begin{subequations}\label{e.106}
\begin{align}
\label{e.106a}
&\mathbb{R}(t,X(t))=\left(\begin{array}{cc}\mathbb{R}_{11}(t,X(t)) & \mathbb{R}_{12}(t,X(t)) \\\mathbb{R}_{12}^\top(t,X(t)) & \mathbb{R}_{22}(t,X(t))\end{array}\right)\\
\label{e.106b}
&\mathbb{R}_{ik}(t,X(t))\triangleq R_{ik}(t)+\sum_{j=1}^r B_{ji}^{\top}(t)X(t)B_{jk}(t),\; i,k=1,2
\end{align}
\end{subequations}
\begin{subequations}\label{e.107}
\begin{align}
\label{e.107a}
&\hat{\mathbb{R}}(t,X(t))=\left(\begin{array}{cc}\hat{\mathbb{R}}_{11}(t,X(t)) & \hat{\mathbb{R}}_{12}(t,X(t)) \\ \hat{\mathbb{R}}_{12}^\top(t,X(t)) & \hat{\mathbb{R}}_{22}(t,X(t))\end{array}\right)\\
\label{e.107b}
&\hat{\mathbb{R}}_{ik}(t,X(t))\triangleq \hat{R}_{ik}(t)+\sum_{j=1}^r \hat{B}_{ji}^{\top}(t)X(t)\hat{B}_{jk}(t),\; i,k=1,2
\end{align}
\end{subequations}
It is known that the solution of a matrix Riccati-type differential equation (RDE) is involved in the construction of the solution of various linear quadratic optimal control problems. The applicability of the solutions of the RDEs (\ref{e.9}) and (\ref{e.10}) to a certain linear quadratic optimal control problem is closely related to the sign of the the matrices introduced via (\ref{e.106})-(\ref{e.107}). Thus, in \cite{Sun:17}, it was shown that if the solutions $X(\cdot)$, $\hat{X}(\cdot)$ to the RDEs (\ref{e.9}) and (\ref{e.10}), respectively, are defined on the whole interval $[0,\; T]$ and additionally, the following inequalities hold:
\begin{subequations}\label{e.108}
\begin{align}
\label{e.108a}
&\mathbb{R}(t,X(t))\geq \gamma I_m\\
\label{e.108b}
&\hat{\mathbb{R}}(t,X(t))\geq \hat{\gamma} I_m
\end{align}
\end{subequations}
for all $t\in [0,\; T]$, $\gamma$, $\hat{\gamma}$ being positive constants, then these solutions are involved in the computation of the gain matrices of the optimal control in a mean-field LQ optimal control problem. Also, \cite{Sun:17} proposed conditions which guarantee the global existence of the solutions of the RDEs of type (\ref{e.9}) and (\ref{e.10}) satisfying sign conditions of type (\ref{e.108}). In \cite{Sun:21bis} the authors provided sufficient conditions which guarantee the global existence of the whole interval $[0,\; T]$ of the solutions $X(\cdot)$, $\hat X(\cdot)$ of the RDEs (\ref{e.9}) and (\ref{e.10}) satisfying sign conditions of the form:
\begin{subequations}\label{e.109}
\begin{align}
\label{e.109a}
&\mathbb{R}_{11}(t,X(t))\geq \mu_1 I_{m_1}\\
\label{e.109b}
&{\mathbb{R}}_{22}(t,X(t))\leq -\mu_2 I_{m_2};\quad \text{for all $t\in [0,\; T]$}
\end{align}
\end{subequations}
and
\begin{subequations}\label{e.110}
\begin{align}
\label{e.110a}
&{\hat{\mathbb{R}}}_{11}(t,X(t))\geq {\hat\mu}_1 I_{m_1}\\
\label{e.110b}
&{\hat{\mathbb{R}}}_{22}(t,X(t))\leq -\hat{\mu}_2 I_{m_2};\quad \text{for all $t\in [0,\; T]$}
\end{align}
\end{subequations}
$\mu_k$, $\hat\mu_k$, $k=1,2$ being positive constants.\\
One shows that in this case the solutions $X(\cdot)$, $\hat X(\cdot)$ of (\ref{e.9}) and (\ref{e.10}), respectively, are involved in the construction of an open-loop saddle point for a mean field zero-sum LQ differential game.\/
Our aim in the present work is to provide a set of necessary and sufficient conditions which guarantee that the solutions $X(\cdot)$, $\hat X(\cdot)$ to RDEs (\ref{e.9}) and (\ref{e.10}) are defined on the whole interval $[0,\; T]$ and satisfy the conditions:
\begin{subequations}\label{e.111}
\begin{align}
\label{e.111a}
&\mathbb{R}_{22}(t,X(t))\geq \nu_2 I_{m_2}\\
\label{e.111b}
&{\mathbb{R}}_{22}^\sharp(t,X(t))\triangleq {\mathbb{R}}_{11}(t,X(t))-{\mathbb{R}}_{12}(t,X(t)){\mathbb{R}}_{22}^{-1}(t,X(t)){\mathbb{R}}_{12}^\top(t,X(t)) \leq -\nu_1 I_{m_1};\quad t\in [0,\; T]
\end{align}
\end{subequations}
and
\begin{subequations}\label{e.112}
\begin{align}
\label{e.112a}
&\hat{\mathbb{R}}_{22}(t,X(t))\geq \hat{\nu}_2 I_{m_2}\\
\label{e.112b}
&\hat{{\mathbb{R}}}_{22}^\sharp(t,X(t))\triangleq {\hat{\mathbb{R}}}_{11}(t,X(t))-{\hat{\mathbb{R}}}_{12}(t,X(t)){\hat{\mathbb{R}}}_{22}^{-1}(t,X(t)){\hat{\mathbb{R}}}_{12}^\top(t,X(t)) \leq -\hat{\nu}_1 I_{m_1};\quad t\in [0,\; T]
\end{align}
\end{subequations}
$\nu_k$, $\hat \nu_k$, $k=1,2$ being positive constants.\\
We shall show that in this case $X(\cdot)$, $\hat{X}(\cdot)$ allow us to construct a state-feedback representation of an open-loop Stackelberg equilibrium for a two person zero-sum mean-field LQ differential game. Also, we shall provide a set of necessary and sufficient conditions which guarantee the global existence of the solutions $X(\cdot)$, $\hat{X}(\cdot)$ of the RDEs (\ref{e.9}) and (\ref{e.10}) which satisfy the sign conditions (\ref{e.111}) and (\ref{e.112}). We will also provide a set of necessary and sufficient conditions which guarantee the global existence of the solutions $X(\cdot)$, $\hat{X}(\cdot)$ of the RDEs (\ref{e.9}) and (\ref{e.10}) which satisfy the sign conditions (\ref{e.109}) and (\ref{e.110}). This would be viewed as an improvement of the results proved in Theorem 4.2 from \cite{Sun:21bis} where only sufficient conditions are provided.
\begin{definition}\label{D2.3}
For a function $H(\cdot): {\cal I}\subset {\mathbb R} \to {\cal S}_n $ we say that:\\
a) $H(\cdot)$ is uniform positive on ${\cal I}$ and we shall write $H(t)\ggcurly0$, $t\in{\cal I}$ if there exists $\mu>0$ such that $H(t)\geq  \mu I_n $  for all $t\in{\cal I}$.\\
b) $H(\cdot)$ is uniform negative on ${\cal I}$ and we shall write $H(t)\llcurly0$, $t\in{\cal I}$ if $-H(\cdot)$ is uniform positive on ${\cal I}$.
\end{definition}
\begin{rem}\label{rem200}
a) If  ${\cal I}\subset {\mathbb R}$ is a compact interval (close and bounded interval), then any continuous function $H(\cdot): {\cal I}\to {\cal S}_n$ is :\\
(i) uniform positive on ${\cal I}$ if and only if $H(t)> 0$ for all $t\in{\cal I}$;\\
(ii) uniform negative on ${\cal I}$ if and only if $H(t)< 0$ for all $t\in{\cal I}$.\\
b) According with the Definition \ref{D2.3} the sign conditions (\ref{e.111})- (\ref{e.112}) may be written as:
\begin{itemize}
\item [$\alpha$)] ${\mathbb R}_{22}(t,{ X}(t))\ggcurly0$, $t\in{\cal I}_X$;
\item [$\beta$)] $ {\mathbb R}_{22}^{\sharp}(t,{X}(t))\llcurly0$, $t\in{\cal I}_X$;
\item [$\gamma$)] $ \hat{{\mathbb R}}_{22}(t,{X}(t))\ggcurly0$, $t\in{\cal I}_X$;
\item[$\delta$)] $ \hat{{\mathbb R}}_{22}^{\sharp}(t,{X}(t))\llcurly0$, $t\in {\cal I}_X$.
\end{itemize}
\end{rem}
\begin{rem}\label{rem.1}
Let ${X}(\cdot): [0,\; T]\to {\cal S}_n$ be a solution of the TVP (\ref{e.9}) which satisfies the sign conditions (\ref{e.111}) and (\ref{e.112}).
Under this condition the corresponding matrix ${\mathbb R}(t,X(t))$ admits the factorization:
\begin{equation}\label{e.15}
{\mathbb R}(t, X(t))={\mathbb V}^\top(t, X(t)) diag (-I_{m_1}, I_{m_2}){\mathbb V}(t, X(t))
\end{equation}
where ${\mathbb V}(t, X(t))=\left(\begin{array}{cc} {\mathbb V}_{11}(t)& 0\\ {\mathbb V}_{21}(t) &{\mathbb V}_{22}(t)\end{array}\right)$.
A possible choice of the block components ${\mathbb V}_{jk}(t)$ is:
\begin{align}\label{e.16}
&{\mathbb V}_{11}(t)=(-{\mathbb R}_{22}^{\sharp}(t,X(t))^{\frac{1}{2}}\notag\\
&{\mathbb V}_{21}(t)=(R_{22}(t)+\sum_{k=1}^r B_{k2}^\top(t) X(t) B_{k2}(t))^{\frac{-1}{2}}(R_{12}(t)+\sum_{k=1}^rB_{k1}^\top(t)X(t)B_{k2}(t))^\top\notag\\
&{\mathbb V}_{22}(t)=(R_{22}(t)+\sum_{k=1}^r B_{k2}^\top(t)X(t) B_{k2}(t))^{\frac{1}{2}}
\end{align}
Similarly
\begin{equation}\label{e.17}
\hat{{\mathbb R}}(t, X(t))=\hat{{\mathbb V}}^\top(t, X(t)) diag (-I_{m_1}, I_{m_2})\hat{{\mathbb V}}(t, X(t))
\end{equation}
where $\hat{{\mathbb V}}(t, X(t))=\left(\begin{array}{cc} \hat{{\mathbb V}}_{11}(t)& 0\\ \hat{{\mathbb V}}_{21}(t) &\hat{{\mathbb V}}_{22}(t)\end{array}\right)$.
A possible choice of the block components $\hat{{\mathbb V}}_{jk}(t)$ is:
\begin{align}\label{e.18}
&\hat{{\mathbb V}}_{11}(t)=(-\hat{{\mathbb R}}_{22}^{\sharp}(t,X(t))^{\frac{1}{2}}\notag\\
&\hat{{\mathbb V}}_{21}(t)=(\hat{R}_{22}(t)+\sum_{k=1}^r \hat{B}_{k2}^\top(t) X(t) \hat{B}_{k2}(t))^{\frac{-1}{2}}(\hat{R}_{12}(t)+\sum_{k=1}^r\hat{B}_{k1}^\top(t)X(t)\hat{B}_{k2}(t))^\top\notag\\
&\hat{{\mathbb V}}_{22}(t)=(\hat{R}_{22}(t)+\sum_{k=1}^r \hat{B}_{k2}^\top(t)X(t) \hat{B}_{k2}(t))^{\frac{1}{2}}
\end{align}
\end{rem}
Let us now define:
\begin{equation}
\begin{cases}
x^1(t)=x(t)-\mathbb{E}[x(t)]\\
x^2(t)=\mathbb{E}[x(t)]\\
u^1(t)=u(t)-\mathbb{E}[u(t)]\\
u^2(t)=\mathbb{E}[u(t)]
\end{cases}
\end{equation}
$t\geq s\geq 0$. The dynamics of the new state variables $x^1(\cdot)$ and $x^2(\cdot)$ are given by:
\begin{equation}\label{e.x1}
\begin{cases}
dx^1(t)=\left\{A_0(t)x^1(t)+B_0(t)u^1(t)\right\}dt\\
\quad+\sum_{j=1}^r\left\{A_j(t)x^1(t)+B_j(t)u^1(t)+\hat{A}_j(t)x^2(t)+\hat{B}_j(t)u^2(t)\right\}dw_j(t)\\
x^1(s)=0
\end{cases}
\end{equation}
\begin{equation}\label{e.x2}
\begin{cases}
dx^2(t)=\left\{\hat{A}_0(t)x^2(t)+\hat{B}_0(t)u^2(t)\right\}dt\\
x^2(s)=x_s\in \mathbb{R}^n
\end{cases}
\end{equation}
\noindent Let $x_u^1(\cdot)$ be the solution of the equation
\begin{align}\label{e.26}
&dx^1(t)=\left\{A_0(t)x^1(t)+B_0(t)u^1(t)\right\}dt\notag\\
&+\sum_{j=1}^r\left\{A_j(t)x^1(t)+B_j(t)u^1(t)+\hat{A}_j(t)x^2_{u^2}(t)+\hat{B}_j(t)u^2(t)\right\}dw_j(t)
\end{align}
with $x_u^1(s)=0$, $x^2_{u^2}(\cdot)$ being the solution of the equation
\begin{equation}\label{e.27}
dx^2(t)=\left\{\hat{A}_0(t)x^2(t)+\hat{B}_0(t)u^2(t)\right\}dt
\end{equation}
with $x^2_{u^2}(s)=x_s\in \mathbb{R}^n$, $t\in[s,\; T]$. Define $\tilde{x}^1(\cdot)$ as the solution of the equation
\begin{align}\label{e.28}
&dx^1(t)=\left[A_0(t)+B_0(t)F(t)\right]x^1(t)dt\notag\\
&+\sum_{j=1}^r\left\{\left[A_j(t)+B_j(t)F(t)\right]x^1(t)+\left[\hat{A}_j(t)+\hat{B}_j(t)\hat{F}(t)\right]\tilde{x}^2(t)\right\}dw_j(t)
\end{align}
with $\tilde{x}^1(s)=0$, $\tilde{x}^2(\cdot)$ being the solution of the equation
\begin{equation}\label{e.29}
dx^2(t)=\left[\hat{A}_0(t)+\hat{B}_0(t)\hat{F}(t)\right]x^2(t)dt
\end{equation}
with $\tilde{x}^2(s)=x_s\in \mathbb{R}^n$, where:
\begin{align}\label{e.22}
F(t)=-\left(R(t)+\sum_{j=1}^rB_j^{\top}(t)X(t)B_j(t)\right)^{-1}\left(B_0^{\top}(t)X(t)+\sum_{j=1}^rB_j^{\top}(t)X(t)A_j(t)+L^{\top}(t)\right)
\end{align}
\begin{align}\label{e.23}
\hat{F}(t)= -\left(\hat{R}(t)+\sum_{j=1}^r\hat{B}_j^{\top}(t)X(t)\hat{B}_j(t)\right)^{-1}
\left(\hat{B}_0^{\top}(t)\hat X(t)+\sum_{j=1}^r\hat{B}_j^{\top}(t)X(t)\hat{A}_j(t)+\hat{L}^{\top}(t)\right)
\end{align}
We denote:
\begin{eqnarray}\label{e.604}
 \tilde{u}^1(t)=\left(\begin{array}{c}\tilde{u}^1_1(t) \\\tilde{u}^1_2(t)\end{array}\right)\triangleq\left(\begin{array}{c}F_1(t)\tilde{x}^1(t) \\F_2(t)\tilde{x}^1(t)\end{array}\right)
 \end{eqnarray}
 \begin{eqnarray}\label{e.605}
  \tilde{u}^2(t)=\left(\begin{array}{c}\tilde{u}^2_1(t)\\ \tilde{u}^2_2(t)\end{array}\right)\triangleq\left(\begin{array}{c}\hat{F}_1(t)\tilde{x}^2(t) \\\hat{F}_2(t)\tilde{x}^2(t)\end{array}\right)
 \end{eqnarray}
$ t\in[s,\; T]$, where: $
\begin{cases}
F_1(t)=\left(\begin{array}{cc}I_{m_1} & 0\end{array}\right)F(t)\\
F_2(t)=\left(\begin{array}{cc}0 & I_{m_2}\end{array}\right)F(t)
\end{cases}
,
\begin{cases}
\hat{F}_1(t)=\left(\begin{array}{cc}I_{m_1} & 0\end{array}\right)\hat{F}(t)\\
\hat{F}_2(t)=\left(\begin{array}{cc}0 & I_{m_2}\end{array}\right)\hat{F}(t)
\end{cases}
$.\\
\noindent Let $\hat{x}^1_{u}(\cdot)$ be the solution of the equation
\begin{align}\label{e.30}
&dx^1(t)=\left\{A_0(t)x^1(t)+B_0(t)(u^1(t)-\tilde{u}^1(t))\right\}dt\notag\\
&+\sum_{j=1}^r\left\{A_j(t)x^1(t)+B_j(t)(u^1(t)-\tilde{u}^1(t))+\hat{A}_j(t)\hat{x}^2_{u^2-\tilde{u}^2}(t)+\hat{B}_j(t)(u^2(t)-\tilde{u}^2(t))\right\}dw_j(t)
\end{align}
with $\hat{x}^1_u(s)=0$, $\hat{x}^2_{u^2-\tilde{u}^2}(\cdot)$ being the solution of the equation
\begin{equation}\label{e.31}
dx^2(t)=\left\{\hat{A}_0(t)x^2(t)+\hat{B}_0(t)(u^2(t)-\tilde{u}^2(t))\right\}dt
\end{equation}
with $\hat{x}^2_{u^2-\tilde{u}^2}(s)=0$, $t\in[s,\; T]$.\\
\noindent The next result will be involved in the proof of the main result of this section.
\begin{lem}\label{L2.1}
Assume that the Riccati equations (\ref{e.9}) and (\ref{e.10}) are solvable on $\mathcal{I}_{{X}}$ and $\mathcal{I}_{\hat{{X}}}$, respectively. Let $F(t)$ and $\hat{F}(t)$ be the feedback gains defined via (\ref{e.22}) and (\ref{e.23}), respectively. Then for every $s\in \mathcal{I}_{{X}}\bigcap \mathcal{I}_{\hat{{X}}}$, the quadratic functional (\ref{e.6}) has the representation:
\begin{align}\label{e.48}
{\mathcal J}(s,T,x_s;u_1(\cdot),u_2(\cdot))&=\frac{1}{2}\langle \hat{X}(s)x_s,x_s\rangle+\frac{1}{2}\mathbb{E}\int_s^T\left[\left(\begin{array}{c}u_1^1(t)-{F}_1(t)x^1_u(t) \\u_2^1(t)-{F}_2(t)x^1_u(t)\end{array}\right)^{\top}\mathbb{R}(t,{{X}}(t))\star\right.\notag\\
&\left.+\left(\begin{array}{c}u_1^2(t)-\hat{F}_1(t)x^2_{u^2}(t) \\u_2^2(t)-\hat{F}_2(t)x^2_{u^2}(t)\end{array}\right)^{\top}\hat{\mathbb{R}}(t,{{X}}(t))\star\right] dt
\end{align}
for all $u(\cdot)=(u_1(\cdot),\;u_2(\cdot))\in \mathcal{U}_{ad}$ and for all $x_s\in L_{\mathcal{F}_s}^2 (\Omega;\; \mathbb{R}^n)$.
\end{lem}
The proof of the above lemma relies on the use of the It\^o formula to systems (\ref{e.26})-(\ref{e.27}) and to the function $v(t,x^1_u,x^2_{u^2})=\langle X(t)x^1_u,x^1_u\rangle+\langle \hat{X}(t)x^2_{u^2},x^2_{u^2}\rangle$, $t\in [s,\; T]$, $x^1_u\in \mathbb{R}^n$, $x^2_{u^2}\in \mathbb{R}^n$ and by taking into account the Riccati equations (\ref{e.9}) and (\ref{e.10}).\\
\noindent Now we are in position to prove the main result of this section.
\begin{thm}\label{T.1}
Assume:
\begin{itemize}
\item [a)] $\mathcal{J}(0,T,0;0,u_2(\cdot))\geq 0$, $\forall u_2(\cdot) \in L_w^2\left\{[0,\; T],\mathbb{R}^{m_2}\right\}$;
\item [b)] the Riccati equations (\ref{e.9}) and (\ref{e.10}) admit solutions $X(\cdot):[0,\; T]\rightarrow \mathcal{S}_n$ and $\hat{X}(\cdot):[0,\; T]\rightarrow \mathcal{S}_n$, respectively, such that $X(\cdot)$ satisfies the sign conditions (\ref{e.111})-(\ref{e.112}) for any $t \in [0;\; T]$.
\end{itemize}
\noindent Let $\tilde{u}_2(\cdot,u_1(\cdot))$ be defined by:
\begin{eqnarray}\label{e.32}
\tilde{u}_2(t,u_1(t))=\tilde u_2^1(t, u_1(t))+\tilde u_2^2(t, u_1(t))
\end{eqnarray}
where
\begin{eqnarray}\label{e.602}
\tilde u_2^1(t,u_1(t))= K(t)x^1_{u_1}(t)+W(t)u^1_1(t)
\end{eqnarray}
\begin{eqnarray}\label{e.603}
\tilde u_2^2(t,u_1(t))=\hat{K}(t)x^2_{u^2_1}(t)+\hat{W}(t)u^2_1(t)
\end{eqnarray}
\noindent $x^1_{u_1}(\cdot)$ being the solution of the initial value problem:
\begin{align}
&dx^1(t)=\left\{\left[A_0(t)+B_{02}(t)K(t)\right]x^1(t)+\left[B_{01}(t)+B_{02}(t)W(t)\right]u^1_1(t)\right\}dt\notag\\
&+\sum_{j=1}^r\left\{\left[A_j(t)+B_{j2}(t)K(t)\right]x^1(t)+\left[B_{j1}(t)+B_{j2}(t)W(t)\right]u^1_1(t)\right.\notag\\
&\left.+\left[\hat{A}_j(t)+\hat{B}_{j2}(t)\hat{K}(t)\right]x^2_{u^2_1}(t)+\left[\hat{B}_{j1}(t)+\hat{B}_{j2}(t)\hat{W}(t)\right]u^2_1(t)\right\}dw_j(t)
\end{align}
$x^1_{u_1}(s)=0$, and $x^2_{u^2_1}(\cdot)$ being the solution of the initial value problem
\begin{equation}
dx^2(t)=\left\{\left[\hat{A}_0(t)+\hat{B}_{02}(t)\hat{K}(t)\right]x^2(t)+\left[\hat{B}_{01}(t)+\hat{B}_{02}(t)\hat{W}(t)\right]u^2_1(t)\right\}dt
\end{equation}
with $x^2_{u^2_1}(s)=x_s\in \mathbb{R}^n$, $t\in[s,\; T]$ and:
\begin{equation}\label{e.24}
\begin{cases}
K(t)=\mathbb{V}_{22}^{-1}(t)\left(\begin{array}{cc} \mathbb{V}_{21}(t) & \mathbb{V}_{22}(t)\end{array}\right)F(t)\\
\hat{K}(t)=\hat{\mathbb{V}}_{22}^{-1}(t)\left(\begin{array}{cc} \hat{\mathbb{V}}_{21}(t) & \hat{\mathbb{V}}_{22}(t)\end{array}\right)\hat{F}(t)
\end{cases}
\end{equation}
\begin{equation}\label{e.25}
\begin{cases}
W(t)=-\mathbb{V}_{22}^{-1}(t) \mathbb{V}_{21}(t) \\
\hat{W}(t)=-\hat{\mathbb{V}}_{22}^{-1}(t)\hat{ \mathbb{V}}_{21}(t)
\end{cases}
\end{equation}
Under these conditions the following hold:
\begin{itemize}
\item [i)] for each $u_1(\cdot)\in L_{w}^2([s,T],\;\mathbb{R}^{m_1})$, $(u_1(\cdot),\; \tilde{{u}}_2(\cdot,u_1(\cdot)))$ lies in ${\mathcal{U}}_{ad}$;
\item[ii)] Let $\tilde{u}_1(\cdot)$ be defined by: $\tilde{u}_1(t)=\tilde{u}_1^1(t)+\tilde{u}_1^2(t)$, $t\in [s,\; T]$, where $\tilde{u}_1^1(\cdot)$ are $\tilde{u}_1^2(\cdot)$ are defined in (\ref{e.604}) and (\ref{e.605}), respectively. Hence, $(\tilde{u}_1(\cdot),\;\tilde{{u}}_2(\cdot,\tilde{u}_1(\cdot)))$ is a Stackelberg  equilibrium  for the two players  zero-sum  mean field LQ differential game described by (5), (6) and the class of admissible strategies ${\mathcal{U}}_{ad}= L_{w}^2([s,T],\;\mathbb{R}^{m_1})\times L_{w}^2([s,T],\;\mathbb{R}^{m_2})$.
\end{itemize}
\end{thm}
\begin{proof}
The assertion in i) is trivial.
Let us now prove the assertion ii).
First, from (\ref{e.604}), (\ref{e.605}), (\ref{e.32}), (\ref{e.602}), (\ref{e.603}), (\ref{e.24}) and (\ref{e.25}) together with the uniqueness of the solution of an initial value problem (IVP)  one obtains that
$\left(\begin{array}{c} \tilde{u}_1^1(t) \\ \tilde{u}_2^1(t, \tilde u_1(t))\end{array}\right) =\left(\begin{array}{c}F_1(t)\tilde{x}^1(t) \\{F}_2(t)\tilde{x}^1(t)\end{array}\right)$  and
$\left(\begin{array}{c} \tilde{u}_1^2(t) \\ \tilde{u}_2^2(t, \tilde u_1^2(t))\end{array}\right) =\left(\begin{array}{c}\hat F_1(t)\tilde{x}^2(t) \\\hat{F}_2(t)\tilde{x}^2(t)\end{array}\right)$.\\
\noindent From (\ref{e.48}) we obtain that:
\begin{equation}\label{th1_e1}
{\mathcal J}(s,T,x_s;\tilde{u}_1(\cdot),\tilde{u}_2(\cdot))=\frac{1}{2}{\mathbb E}[\langle \hat{X}(s)x_s,x_s\rangle]
\end{equation}
for all $x_s\in \mathbb{R}^n$.
On the other hand, using the factorizations (\ref{e.15})-(\ref{e.16}) and (\ref{e.17})-(\ref{e.18}) we rewrite (\ref{e.48}) in the form:
\begin{align}\label{e5.19}
&\mathcal{J}(s,T,x_s;u_1(\cdot),u_2(\cdot))=\frac{1}{2}{\mathbb E}[\langle \hat{X}(s)x_s,x_s\rangle]\notag\\
&+\frac{1}{2}\mathbb{E}\Big[\int_s^T\big(|{\mathbb{V}}_{21}(t)(u^1_1(t)-F_1(t)x^1_u(t))+{\mathbb{V}}_{22}(t)(u^1_2(t)-F_2(t)x^1_u(t))|^2\notag\\
&-|{\mathbb{V}}_{11}(t)(u^1_1(t)-F_1(t)x^1_u(t))|^2\big)dt\Big]\notag\\
&+\frac{1}{2}\mathbb{E}\Big[\int_s^T\big(|{\hat{\mathbb{V}}}_{21}(t)(u^2_1(t)-\hat{F}_1(t)x^2_{u^2}(t))+{\hat{\mathbb{V}}}_{22}(t)(u^2_2(t)-\hat{F}_2(t)x^2_{u^2}(t))|^2\notag\\
&-|{\hat{\mathbb{V}}}_{11}(t)(u^2_1(t)-\hat{F}_1(t)x^2_{u^2}(t))|^2\big)dt\Big]
\end{align}
It follows from the uniqueness of the solution of an initial value problem (IVP) that:
\begin{align}\label{e5.20}
&\inf_{u_2(\cdot)\in{\cal U}_2}\mathcal{J}(s,T,x_s;u_1(\cdot),u_2(\cdot))= {\cal J}(s,T,x_s; u_1(\cdot), \tilde u_2(\cdot, u_1(\cdot)))\notag\\
&=\frac{1}{2}{\mathbb E}[\langle \hat{X}(s)x_s,x_s\rangle] -\frac{1}{2}\mathbb{E}\Big[\int_s^T|{\mathbb{V}}_{11}(t)(u^1_1(t)-F_1(t)x^1_{u_1}(t))|^2 dt\Big]\notag\\
&-\frac{1}{2}\mathbb{E}\Big[\int_s^T|{\hat{\mathbb{V}}}_{11}(t)(u^2_1(t)-\hat{F}_1(t)x^2_{u^2_1}(t))|^2 dt\Big]
\end{align}
for all $u_1(\cdot)\in {L}_{ w}^2([s,T],\;\mathbb{R}^{m_1})$.
From (\ref{e5.20}) and (\ref{th1_e1}) we obtain that:
$$\sup_{u_1(\cdot)\in{\cal U}_1} \inf_{u_2(\cdot)\in{\cal U}_2} \mathcal{J}(s,T,x_s;u_1(\cdot),u_2(\cdot)) =
\mathcal{J}(s,T,x_s;\tilde u_1(\cdot),\tilde{u}_2(\cdot,\tilde u_1(\cdot)))$$
 which confirms the validity of the first equality from (\ref{e.7}).
On the other hand, from Lemma 5.1 we obtain that
$$\mathcal{J}(s,T,x_s;{\tilde u}_1(\cdot),{\tilde{u}}_2(\cdot,{\tilde{u}}_1(\cdot))) \leq \mathcal{J}(s,T,x_s;{\tilde{u}}_1(\cdot), u_2(\cdot))$$
for all $u_2(\cdot) \in L_w^2([s,T]; {\mathbb R}^{m_2})$.
Hence, $$\mathcal{J}(s,T, x_s;{\tilde{u}}_1(\cdot),{\tilde{u}}_2(\cdot,{\tilde{u}}_1(\cdot)))=\inf_{u_2(\cdot)\in{\cal U}_2}
\mathcal{J}(s,T,x_s;{\tilde{u}}_1(\cdot),u_2(\cdot))$$
because $\tilde u_2(\cdot, \tilde u_1(\cdot))\in{\cal U}_2$.
So we have shown that the second equality from  (\ref{e.7}) holds true. Thus the proof is complete.
\end{proof}
\section{The Riccati equations}

\subsection{Several preliminary results.}
In this subsection we deal first with the problem of the global existence on the whole interval $[0,\; T]$ of the solutions $X(\cdot)$, $\hat X(\cdot)$ of the TVPs (\ref{e.9}) and (\ref{e.10}) respectively, when their quadratic terms are of definite sign.
This result will be used later to state the necessary and sufficient conditions for the global existence of the solutions of the TVPs (\ref{e.9}) and (\ref{e.10}) satisfying conditions (\ref{e.111}), (\ref{e.112}).
The next result is a compact version of the result stated in Lemma \ref{L2.1} adapted to the use in this section.
\begin{lem}\label{L3.1}
If $s\in \mathcal{I}_X\bigcap \mathcal{I}_{\hat X}$ then for any $u(\cdot)\in L_w^2([s, T],\mathbb{R}^m)$ and $x_s \in {\mathbb R}^n$ we have:
\begin{align}
\mathcal{J}(s,T,x_s;u(\cdot))&=\frac{1}{2}\left\langle \hat X(s)x_s, x_s\right\rangle\notag\\
&+\frac{1}{2}{\mathbb E} \left[ \int_s^T \left\{\left\langle\left( R(t)+\sum_{j=1}^rB_j^\top(t)X(t)B_j(t)\right)\left(u^1(t)-F(t)x_u^1(t)\right),u^1(t)-F(t)x_u^1(t)\right\rangle
\right.\right.\notag\\
&+\left.\left.\left\langle\left(\hat{R}(t)+\sum_{j=1}^r\hat{B}_j^\top(t)X(t)\hat{B}_j(t)\right)
\left(u^2(t)-\hat{F}(t)x_{u^2}^2(t)\right),u^2(t)-F(t)x_{u^2}^2(t)\right\rangle\right\}dt\right]
\end{align}
where $F(t)$ and $\hat F(t)$ were defined in (\ref{e.22}), (\ref{e.23}), respectively.
\end{lem}
The next two propositions will play an important role in the proof of the main result of this section.
\begin{prop}\label{P3.2}
For the controlled system (\ref{e.1}) and the quadratic functional (\ref{e.3}) the following are equivalent:
\begin{itemize}
\item [i)] there exists $\delta >0$ with the property that:
\begin{equation}\label{e.113}
\mathcal{J}(0,T,0;u(\cdot))\geq \delta \mathbb{E}\left[\int_0^T|u(t)|^2 dt\right];\quad \forall u(\cdot) \in L_w^2([0,\; T],\mathbb{R}^m)
\end{equation}
\item [ii)] the TVPs (\ref{e.9}) and (\ref{e.10}) have the solutions $X(\cdot):[0,\; T]\rightarrow \mathcal{S}_n$ and $\hat X(\cdot):[0,\; T]\rightarrow \mathcal{S}_n$, respectively, which are satisfying the sign conditions:
\begin{equation}\label{e.114}
R(t)+\sum_{j=1}^r B_j^\top(t) X(t)B_j(t)\geq \gamma I_m
\end{equation}
\begin{equation}\label{e.115}
\hat{R}(t)+\sum_{j=1}^r \hat{B}_j^\top(t) X(t)\hat{B}_j(t)\geq \hat{\gamma} I_m
\end{equation}
for all $t\in [0,\; T]$, $\gamma$ and $\hat{\gamma}$ being positive constants.
\end{itemize}
\end{prop}
\begin{proof}
The implication $i)\Rightarrow ii)$ follows directly applying Theorem 4.2 and Theorem 4.4 from \cite{Sun:17} in the case of the TVPs (\ref{e.9}) and (\ref{e.10}), respectively.\\
Let us now prove the implication $ii)\Rightarrow i)$. First we show that:
\begin{equation}\label{e.118}
\delta \triangleq\inf\left\{\mathcal{J}(0,T,0;u(\cdot))|u(\cdot)\in L_w^2([0,\; T];\mathbb{R}^m)\; \text{with}\; \|u(\cdot)\|_2=1\right\}>0
\end{equation}
where we denoted:
\begin{equation}\label{e.119}
\|u(\cdot)\|_2\triangleq \left({\mathbb E}\left[\int_0^T|u(t)|^2dt\right]dt\right)^{\frac{1}{2}}
\end{equation}
Let us assume by contrary that (\ref{e.118}) is not true. In this case, there exists a sequence $\{u_k(\cdot)\}_{k\in \mathbb{Z}_+}\subset L_w^2\left([0,\; T];\mathbb{R}^m\right)$ such that $\|u_k(\cdot)\|_2=1$ and
\begin{equation}\label{e.120}
\underset{k\rightarrow \infty}{\lim} \mathcal{J}(0,T,0;u_k(\cdot))=0
\end{equation}
Let $\left(x_k^1(\cdot),\; x^2_k(\cdot)\right)$ be the solution of the stochastic differential equations (\ref{e.x1}) and (\ref{e.x2}) determined by $u_k^1(\cdot)\triangleq u_k(\cdot)-\mathbb{E}[u_k(\cdot)]$, $u_k^2(\cdot)\triangleq {\mathbb E}[u_k(\cdot)]$, and having the initial conditions $x_k^1(0)=0$, $x_k^2(0)=0$.
We rewrite the equations (\ref{e.x1}) and (\ref{e.x2}) satisfied by $\left(x_k^1(\cdot),\; x^2_k(\cdot)\right)$ as:
\begin{subequations}\label{e.121}
\begin{align}
\label{e.121a}
&d x_k^1(t)=\left[\left(A_0(t)+B_0(t)F(t)\right)x_k^1(t)+f_{0k}(t)\right]dt+\sum_{j=1}^r\left[\left(A_j(t)+B_j(t)F(t)\right)x_k^1(t)+\right.\notag\\
&\left.+\left(\hat{A}_j(t)+\hat{B}_j(t)\hat{F}(t)\right)x_k^2(t)+f_{jk}(t)+g_{jk}(t)\right]dw_j(t)\\
\label{e.121b}
&d x_k^2(t)=\left[\left(\hat{A}_0(t)+\hat{B}_0(t)\hat{F}(t)\right)x_k^2(t)+g_{0k}(t)\right]dt
\end{align}
\end{subequations}
where:
\begin{equation*}
\begin{cases}
f_{jk}(t)\triangleq B_j(t)\left(u_k^1(t)-F(t)x_k^1(t)\right)\\
g_{jk}(t)\triangleq \hat{B}_j(t)\left(u_k^2(t)-\hat{F}(t)x_k^2(t)\right)
\end{cases}
\end{equation*}
$0\leq j\leq r$, $k\in \mathbb{Z}_+$. From Lemma \ref{L3.1} together with (\ref{e.114}), (\ref{e.115}), (\ref{e.119}), (\ref{e.120}) we deduce that $\underset{k\rightarrow \infty}{\lim}\|f_{jk}(\cdot)\|_2=0$ and\\ $\underset{k\rightarrow \infty}{\lim}\|g_{jk}(\cdot)\|_2=0$ for all $0\leq j\leq r$. Further, applying Theorem 3.6.1 from \cite{carte2013} in the case of the solutions of (\ref{e.121}) we may deduce that $\underset{k\rightarrow \infty}{\lim}\mathbb{E}\left[\int_0^T|x_k^i(t)|^2dt\right]=0$, $i=1,2$. Hence:
\begin{align*}
\underset{k\rightarrow \infty}{\lim}\|u_k(\cdot)\|_2^2&=\underset{k\rightarrow \infty}{\lim}(\|u_k^1(\cdot)\|_2^2+\|u_k^2(\cdot)\|_2^2)\leq \underset{k\rightarrow \infty}{\lim}\mathbb{E}\left[\int_0^T|F(t)x_k^1(t)|^2dt\right]+\notag\\
&+\underset{k\rightarrow \infty}{\lim}\mathbb{E}\left[\int_0^T|u^1(k)-F(t)x_k^1(t)|^2dt\right]+\underset{k\rightarrow \infty}{\lim}\mathbb{E}\left[\int_0^T|\hat{F}(t)x_k^2(t)|^2dt\right]\notag\\
&+\underset{k\rightarrow \infty}{\lim}\mathbb{E}\left[\int_0^T|u^2(k)-\hat{F}(t)x_k^2(t)|^2dt\right]=0
\end{align*}
This is not possible because $\|u_k(\cdot)\|_2=1$. Thus (\ref{e.118}) is true. Now, (\ref{e.113}) is obtained from (\ref{e.118}) written for $u(\cdot)$ replaced by $\frac{1}{\|u(\cdot)\|_2}u(\cdot)$ and taking into account that: $$\mathcal{J}\left(0,T,0;\frac{1}{\|u(\cdot)\|_2}u(\cdot)\right)=\frac{1}{\|u(\cdot)\|_2^2}\mathcal{J}(0,T,0;u(\cdot))$$ Hence, we have shown that the implication $ii)\rightarrow i)$ holds. Thus the proof is completed.
\end{proof}
\begin{rem}
The proof of Proposition 3.2 could be viewed as an alternative proof of Theorem 5.2 from \cite{Sun:17} which is of interest on its own. 
\end{rem}
Applying the result proved in Proposition \ref{P3.2} in the case of the quadratic functional $\check{\mathcal{J}}(0,T,x_s;u(\cdot))\triangleq -{\mathcal{J}}(0,T,x_s;u(\cdot))$ we obtain:
\begin{prop}\label{P3.3}
For the controlled system (\ref{e.1}) and the quadratic functional (\ref{e.3}) the following are equivalent:
\begin{itemize}
\item [i)] there exists $\hat{\delta} >0$ with the property that :
\begin{equation}\label{e.133}
\mathcal{J}(0,T,0;u(\cdot))\leq -\hat{\delta} \mathbb{E}\left[\int_0^T|u(t)|^2 dt\right];\quad \forall u(\cdot) \in L_w^2\left([0,\; T],\mathbb{R}^m\right)
\end{equation}
\item [ii)] the TVPs (\ref{e.9}) and (\ref{e.10}) have the solutions $X(\cdot):[0,\; T]\rightarrow \mathcal{S}_n$ and $\hat X(\cdot):[0,\; T]\rightarrow \mathcal{S}_n$, respectively, which are satisfying the sign conditions:
\begin{equation}\label{e.134}
R(t)+\sum_{j=1}^r B_j^\top(t) X(t)B_j(t)\leq -\gamma I_m
\end{equation}
\begin{equation}\label{e.135}
\hat{R}(t)+\sum_{j=1}^r \hat{B}_j^\top(t) X(t)\hat{B}_j(t)\leq -\hat{\gamma} I_m
\end{equation}
for all $t\in [0,\; T]$, $\gamma$ and $\hat{\gamma}$ being positive constants.
\end{itemize}
\end{prop}
\subsection{Necessary and sufficient conditions for the global existence of the TVPs (\ref{e.9})-(\ref{e.10}) subject to constraints (\ref{e.111})-(\ref{e.112})}
In this section we deal with the problem of the global existence, that is, "the existence on the whole interval $[0,\; T]$" of the solution $X(\cdot)$ of the TVP (\ref{e.9}) satisfying the sign conditions (\ref{e.111}) and (\ref{e.112}) and the problem of global existence on  $[0,\; T]$ of the solution $\hat X(\cdot)$ of the TVP (\ref{e.10}).
To this end, first, we shall associate to the pair formed by the controlled system (\ref{e.5}) and the performance criterion (\ref{e.6}), two other adequately defined pairs of controlled systems and associated quadratic functionals.
For these two pairs we shall apply the results proved in Proposition \ref{P3.2} and Proposition \ref{P3.3} in order to derive necessary and sufficient conditions for the global existence of the solutions of (\ref{e.9}) and (\ref{e.10}) satisfying the sign conditions (\ref{e.111}) and (\ref{e.112}).\\
Let $\left(K(\cdot), W(\cdot)\right):[0,\; T]\rightarrow \mathbb{R}^{m_2\times n}\times \mathbb{R}^{m_2\times m_1}$ and $\left(\hat{K}(\cdot),\hat{W}(\cdot)\right):[0,\; T]\rightarrow \mathbb{R}^{m_2\times n}\times \mathbb{R}^{m_2\times m_1}$ be continuous matrix valued functions.
We set:
\begin{equation}\label{e.136}
u_2(t):=K(t)x(t)+W(t)u_1(t)+\left(\hat K(t)-K(t)\right)\mathbb{E}[x(t)]+\left(\hat W(t)-W(t)\right)\mathbb{E}[u_1(t)]
\end{equation}
One sees that:
\begin{equation}\label{e.137}
\mathbb{E}[u_2(t)]=\hat K(t)\mathbb{E}[x(t)]+\hat W(t)\mathbb{E}[u_1(t)]
\end{equation}
The substitution of (\ref{e.136}) and (\ref{e.137}) in (\ref{e.5}) and (\ref{e.6}) yield:
\begin{align}\label{e.138}
&dx(t)=\left(A_{0K}(t)x(t)+A_{0K\hat K}(t)\mathbb{E}[x(t)]+B_{0W}(t)u_1(t)+B_{0W\hat W}(t)\mathbb{E}[u_1(t)]\right)dt+\notag\\
&+\sum_{j=1}^r\left(A_{jK}(t)x(t)+A_{jK\hat K}(t)\mathbb{E}[x(t)]+B_{jW}(t)u_1(t)+B_{jW\hat W}(t)\mathbb{E}[u_1(t)]\right)dw_j(t)\notag \\
&x(s)=x_s \in \mathbb{R}^n
\end{align}
and:
\begin{align}\label{e.139}
&\mathcal{J}_{KW}^{\hat K \hat W}(s,T,x_s;u_1(\cdot))=\frac{1}{2}\mathbb{E}\left[\left \langle G_Tx_{u_1}(T),x_{u_1}(T)\right \rangle+\left \langle \hat{G}_T\mathbb{E}[x_{u_1}(T)],\mathbb{E}[x_{u_1}(T)]\right \rangle+\right.\notag\\
&\left.+\int_s^T\left(\left \langle \left(\begin{array}{cc}M_K(t) & L_{KW}(t) \\\star & R_W(t)\end{array}\right)\left(\begin{array}{c}x_{u_1}(t) \\u_1(t)\end{array}\right),\left(\begin{array}{c}x_{u_1}(t) \\u_1(t)\end{array}\right) \right\rangle+\right.\right.\notag\\
&\left.\left.+\left \langle \left(\begin{array}{cc}M_{K\hat K}(t) & L_{KW}^{\hat K\hat W}(t) \\\star & R_{W\hat W}(t)\end{array}\right)\left(\begin{array}{c}\mathbb{E}[x_{u_1}](t) \\\mathbb{E}[u_1(t)]\end{array}\right),\left(\begin{array}{c}\mathbb{E}[x_{u_1}(t)] \\\mathbb{E}[u_1(t)]\end{array}\right) \right\rangle\right)\right]
\end{align}
where $x_{u_1}(\cdot)=x(\cdot;s,x_s,u_1(\cdot))$ is the solution of the IVP (\ref{e.138}) determined by the input $u_1(\cdot)\in L_w^2\left\{[0,\; T];\mathbb{R}^{m_1}\right\}$. In (\ref{e.138}) and (\ref{e.139}) we have used the notations:
\begin{subequations}\label{e.140}
\begin{align}
\label{e.140a}
&A_{jK\hat K}(t)=\hat A_{j\hat K}(t)-A_{jK}(t)\\
\label{e.140b}
&\hat A_{j\hat K}(t)={\hat A}_j(t)+ {\hat B}_{j2}(t)\hat K(t)\\
\label{e.140c}
& A_{jK}(t)={A}_j(t)+ {B}_{j2}(t) K(t)
\end{align}
\end{subequations}
\begin{subequations}\label{e.141}
\begin{align}
\label{e.141a}
&B_{jW\hat W}(t)=\hat B_{j\hat W}(t)-B_{jW}(t)\\
\label{e.141b}
&\hat B_{j\hat W}(t)={\hat B}_{j1}(t)+ {\hat B}_{j2}(t)\hat W(t)\\
\label{e.141c}
& B_{jW}(t)={B}_{j1}(t)+ {B}_{j2}(t) W(t)
\end{align}
\end{subequations}
$0\leq j\leq r$,
\begin{subequations}\label{e.142}
\begin{align}
\label{e.142a}
&M_{K\hat K}(t)=\hat M_{\hat K}(t)-M_{K}(t)\\
\label{e.142b}
&\hat M_{\hat K}(t)=\left(\begin{array}{c}I_n \\\hat K(t)\end{array}\right)^\top\left(\begin{array}{cc}\hat M(t) & {\hat L}_2(t) \\\star & {\hat R}_{22}(t)\end{array}\right)\left(\begin{array}{c}I_n \\\hat K(t)\end{array}\right)\\
\label{e.142c}
& M_{K}(t)=\left(\begin{array}{c}I_n \\K(t)\end{array}\right)^\top\left(\begin{array}{cc}M(t) & {L}_2(t) \\\star & {R}_{22}(t)\end{array}\right)\left(\begin{array}{c}I_n \\K(t)\end{array}\right)
\end{align}
\end{subequations}
\begin{subequations}\label{e.143}
\begin{align}
\label{e.143a}
&L_{KW}^{\hat K \hat W}(t)={\hat L}_{\hat K \hat W}(t)-L_{KW}(t)\\
\label{e.143b}
&{\hat L}_{\hat K \hat W}(t)={\hat L}_1(t)+{\hat L}_2(t)\hat W(t)+{\hat K}^\top(t){\hat R}_{12}^\top(t)+{\hat K}^\top(t){\hat R}_{22}(t)\hat W(t)\\
\label{e.143c}
& L_{KW}(t)={L}_1(t)+{L}_2(t)W(t)+{K}^\top(t){R}_{12}^\top(t)+{K}^\top(t){R}_{22}(t)W(t)
\end{align}
\end{subequations}
\begin{subequations}\label{e.144}
\begin{align}
\label{e.144a}
&R_{W\hat W}(t)=\hat R_{\hat W}(t)-R_{W}(t)\\
\label{e.144b}
&\hat R_{\hat W}(t)=\left(\begin{array}{c}I_{m_1} \\\hat W(t)\end{array}\right)^\top\left(\begin{array}{cc}{\hat R}_{11}(t) & {\hat R}_{12}(t) \\\star & {\hat R}_{22}(t)\end{array}\right)\left(\begin{array}{c}I_{m_1} \\\hat W(t)\end{array}\right)\\
\label{e.144c}
& R_{W}(t)=\left(\begin{array}{c}I_{m_1} \\ W(t)\end{array}\right)^\top\left(\begin{array}{cc}{ R}_{11}(t) & {R}_{12}(t) \\\star & {R}_{22}(t)\end{array}\right)\left(\begin{array}{c}I_{m_1} \\W(t)\end{array}\right)
\end{align}
\end{subequations}
The Riccati differential equations of type (\ref{e.9}) and (\ref{e.10}), respectively, associated to the pair consisting of the controlled system (\ref{e.138}) and the quadratic functional (\ref{e.139}) are:
\begin{subequations}\label{e.145}
\begin{align}
\label{e.145a}
&-{\dot{Y}}(t)=A_{0K}^{\top}(t)Y(t)+Y(t)A_{0K}(t)+
\sum_{j=1}^rA_{jK}^{\top}(t)Y(t)A_{jK}(t)-\notag\\
&-\Big(Y(t)B_{0W}(t)+\sum_{j=1}^rA_{jK}^{\top}(t)Y(t)B_{jW}(t)+L_{K W}(t)\Big)\Big(R_{W}(t)+\sum_{j=1}^rB_{jW}^{\top}(t)Y(t)B_{jW}(t)\Big)^{-1}\times\notag\\
&\times\Big(B_{0 W}^{\top}(t)Y(t)+\sum_{j=1}^rB_{j W}^{\top}(t)Y(t)A_{jK}(t)+L_{K W}^{\top}(t)\Big)+M_{K}(t)\\
\label{e.145b}
&Y(T)=G_T
\end{align}
\end{subequations}
\begin{subequations}\label{e.146}
\begin{align}
\label{e.146a}
&-{\dot{\hat Y}}(t)={\hat A}_{0\hat K}^{\top}(t)\hat Y(t)+\hat Y(t){\hat A}_{0\hat K}(t)+
\sum_{j=1}^r{\hat A}_{j\hat K}^{\top}(t)Y(t){\hat A}_{j\hat K}(t)-\notag\\
&-\Big(\hat Y(t){\hat B}_{0\hat W}(t)+\sum_{j=1}^r{\hat A}_{j\hat K}^{\top}(t)Y(t){\hat B}_{j\hat W}(t)+{\hat L}_{\hat K \hat W}(t)\Big)\Big({\hat R}_{\hat W}(t)+\sum_{j=1}^r{\hat B}_{j\hat W}^{\top}(t)Y(t){\hat B}_{j\hat W}(t)\Big)^{-1}\times\notag\\
&\times\Big({\hat B}_{0 \hat W}^{\top}(t)\hat Y(t)+\sum_{j=1}^r{\hat B}_{j \hat W}^{\top}(t)Y(t){\hat A}_{j\hat K}(t)+{\hat L}_{\hat K \hat W}^{\top}(t)\Big)+{\hat M}_{\hat K}(t)\\
\label{e.146b}
&\hat Y(T)={\hat G}_T
\end{align}
\end{subequations}
Applying Proposition 3.3 in the case of the pair consisting of the controlled system (\ref{e.138}) and the quadratic functional (\ref{e.139}), we obtain:
\begin{cor}\label{C3.5}
For the controlled system (\ref{e.138}) and the quadratic functional (\ref{e.139}) associated to the pairs of continuous matrix valued functions $\left(K(\cdot),W(\cdot)\right)$ and $\left(\hat K(\cdot),\hat W(\cdot)\right)$, the following are equivalent:
\item [i)] There exists $\hat \delta >0$ with the property that:
\begin{equation}\label{e.147}
\mathcal{J}_{KW}^{\hat K \hat W}(0,T,0;u_1(\cdot))\leq -\hat \delta \mathbb{E}\left[\int_0^T|u_1(t)|^2dt\right]
\end{equation}
$\forall u_1(\cdot)\in L_w^2\left([0,\; T];\mathbb{R}^{m_1}\right)$.
\item [ii)] The TVPs (\ref{e.145}), (\ref{e.146}) have the solutions $Y_{KW}:[0,\; T]\rightarrow \mathcal{S}_n$ and $Y_{KW}^{\hat K \hat W}:[0,\; T]\rightarrow \mathcal{S}_n$ with the additional property that $Y_{KW}(t)$ satisfies the constraints:
\begin{equation}\label{e.148}
R_W(t)+\sum_{j=1}^rB_{j W}^\top(t)Y_{KW}(t)B_{j W}(t)\leq -\mu I_{m_1}
\end{equation}
\begin{equation}\label{e.149}
{\hat R}_{\hat W}(t)+\sum_{j=1}^r{\hat B}_{j \hat W}^\top(t)Y_{KW}(t){\hat B}_{j \hat W}(t)\leq -\hat\mu I_{m_1}
\end{equation}
for all $t\in [0,\; T]$, $\mu$, $\hat\mu$ being positive constants.
\end{cor}
\noindent Applying Lemma \ref{L3.1} in the case of the quadratic functional (\ref{e.139}) for which (\ref{e.147}) holds, we obtain via (\ref{e.148})- (\ref{e.149}) that
\begin{eqnarray}\label{3.1001}
\mathcal{J}_{KW}^{\hat K \hat W}(s,T,x_s;u_1(\cdot))\leq \frac{1}{2}\langle Y_{KW}^{\hat K\hat W}(s)x_s, x_s\rangle
\end{eqnarray}
for any $s\in[0, T) $, $u_1(\cdot)\in L_w^2([s, T], {\mathbb R}^{m_1}), x_s\in{\mathbb R}^n$.\\
\noindent Now we introduce the following condition:\\
\noindent \textit{\textbf{C}1.) There exist continuous matrix valued functions $\left(K(\cdot),W(\cdot)\right):[0,\; T]\rightarrow \mathbb{R}^{m_2\times n}\times \mathbb{R}^{m_2\times m_1}$ and $\left(\hat K(\cdot),\hat W(\cdot)\right):[0,\; T]\rightarrow \mathbb{R}^{m_2\times n}\times \mathbb{R}^{m_2\times m_1}$ with the property that the mapping $u_1(\cdot)\rightarrow \mathcal{J}_{KW}^{\hat K \hat W}(0,T,0;u_1(\cdot)):L_w^2\left\{[0,\; T];\mathbb{R}^{m_1}\right\}\rightarrow \mathbb{R}$ is uniformly concave.
This means that there exists $\hat \delta >0$ for which (\ref{e.147}) holds.}\\
\noindent Consider the following SDE:
\begin{align}\label{e1.3}
dx(t)&=(A_0(t)x(t)+B_{01}(t)\mu_1(t)+ B_{02}(t)\mu_2(t))dt\notag\\
&+\sum_{j=1}^r(A_j(t)x(t)+B_{j1}(t)\mu_1(t)+B_{j2}(t)\mu_2(t))dw_j(t),\quad x(s)=x_s\in{\mathbb{R}}^n
\end{align}
and the cost functional:
\begin{equation}\label{e1.4}
\mathring{\mathcal{J}}(s,T,x_s,\mu_1,\mu_2)=\frac{1}{2}\mathbb{E}\left[\left\langle G_Tx_\mu(T),x_\mu(T)\right\rangle+\int_s^T\left\langle \mathbb{M}(t)\left(\begin{array}{c}x_\mu(t) \\\mu_1(t) \\\mu_2(t)\end{array}\right), \left(\begin{array}{c}x_\mu(t) \\\mu_1(t) \\\mu_2(t)\end{array}\right)\right\rangle dt\right]
\end{equation}
where $\mathbb{M}(t)=\left(\begin{array}{ccc}M(t) & L_1(t) & L_2(t) \\\star & R_{11}(t) & R_{12}(t) \\\star & \star & R_{22}(t)\end{array}\right)$, $0\leq t\leq T$ and $x_\mu(t)$, $0\leq t\leq T$, is the solution of the initial value problem ({\ref{e1.3}}) corresponding to the input $\mu(t)=\left(\begin{array}{c}\mu_1(t) \\\mu_2(t)\end{array}\right)$.\\
\noindent Setting $\mu_2(t)\equiv\mu_{2KW}(t)= K(t)x(t)+W(t)\mu_1(t)$ in ({\ref{e1.3}}) and ({\ref{e1.4}}), we obtain:
\begin{align}\label{e3.4}
dx(t)=(A_{0K}(t)x(t)+B_{0W}(t)\mu_1(t))dt+\sum\limits_{j=1}^r(A_{jK}(t)x(t)+B_{jW}(t)\mu_1(t))dw_j(t),\;
 x(s)=x_s\in{\mathbb{R}}^n
\end{align}
\begin{align}\label{e3.5}
\mathring{\mathcal{J}}_{K W}(s,T,x_s,\mu_1)=\frac{1}{2}\mathbb{E}\left[\left\langle G_Tx_{\mu_1}(T),x_{\mu_1}(T)\right\rangle+\int\limits_0^T\left\langle\left(
                                                \begin{array}{cc}
                                                  M_{K}(t) & L_{K W}(t) \\
                                                  \star & R_{W}(t) \\
                                                \end{array}
                                              \right)\left(
                                       \begin{array}{c}
                                         x_{\mu_1}(t) \\
                                         \mu_{1}(t) \\
                                       \end{array}
                                     \right),\left(
                                       \begin{array}{c}
                                         x_{\mu_1}(t) \\
                                         \mu_{1}(t) \\
                                       \end{array}
                                     \right)\right\rangle dt\right]
\end{align}
where $x_{\mu_1}(t)$ is the solution of the initial value problem ({\ref{e3.4}}) corresponding to the input $\mu_1(t)$.
\begin{lem}\label{L3.6}
a)  Assume that for the continuous matrix valued functions $(K(\cdot), W(\cdot)): [0,T]\to {\mathbb R}^{m_2\times n} \times {\mathbb R}^{m_2\times m_1}$
the solution $Y_{KW}(\cdot)$ of he TVP of type  (\ref{e.145}) is defined on the whole interval $[0,T]$  and satisfies the sign condition (\ref{e.148}).
Under  these conditions we have:
\begin{eqnarray}\label{3.1002}
\mathring{\mathcal{J}}_{K W}(s,T,x_s;\mu_1(\cdot))\leq \frac{1}{2} \langle Y_{KW}(s)x_s, x_s\rangle
\end{eqnarray}
for all $s\in [0, T)$, $\mu_1(\cdot)\in L_w^2([s,T]; {\mathbb R}^{m_1}), x_s\in{\mathbb R}^n$.\\
b) If the solution $X(\cdot)$ of the TVP (\ref{e.9}) is defined and satisfies the sign condition (\ref{e.111}) on an interval $(\tau, T]$  and if $Y_{KW}(\cdot): [0,T] \to {\cal S}_n$ is the solution of the TVP (\ref{e.145}) satisfying the sign condition (\ref{e.148}), then
\begin{eqnarray}\label{3.1003}
X(s)\leq Y_{KW}(s); \quad \forall   s\in (\tau, T].
\end{eqnarray}
\end{lem}
\begin{proof}
a)  Applying It\^o's formula in the case of the function $v(t,x)=x^\top Y_{KW}(t)x $ and to the stochastic process $x_{\mu_1}(t)$ defined as the solution of the IVP (\ref{e3.4}) we obtain:
\begin{align*}
&\mathring{\mathcal{J}}_{KW}(s,T,x_s;\mu_1(\cdot))=\frac{1}{2}\left\langle Y_{KW}(s)x_s, x_s\right\rangle +\notag\\
&+\frac{1}{2} {\mathbb E}\left[\int_s^T \left\langle (R_W(t)+\sum_{j=1}^r B_{jW}^\top(t)Y_{KW}(t)B_{jW}(t))(\mu_1(t)-F_{KW}(t)x_{\mu_1}(t)), \mu_1(t)-F_{KW}(t)x_{\mu_1}(t)\right\rangle dt\right] \notag\\
&\leq \frac{1}{2}\left\langle Y_{KW}(s)x_s, x_s\right\rangle
\end{align*}
where $F_{KW}(t)=-\left(R_W(t)+\sum_{j=1}^r B_{jW}^\top(t)Y_{KW}(t)B_{jW}(t)\right)^{-1}\left(B_{0W}^\top(t)Y_{KW}(t)+\right.\\\left.
+\sum_{j=1}^r B_{jW}^\top(t)Y_{KW}(t)A_{jK}(t)+L_{KW}^\top(t)\right)$. This confirm the validity of (\ref{3.1002}).\\
\noindent b)  Applying It\^o's formula to the function $v_1(t,x)=x^\top X(t)x$  and to the stochastic process $x_{\mu}(t)$ defined as a solution of the IVP
(\ref{e1.3}) we obtain:
\begin{align*}
\mathring{\mathcal{J}}(s,T,x_s;\mu_1(\cdot),\mu_2(\cdot))&=\frac{1}{2}\langle X(s)x_s,x_s\rangle +\notag\\
&+\frac{1}{2} {\mathbb E}\left[\int_s^T \left(\begin{array}{c} \mu_1(t)-F_1(t)x_{\mu}(t)\\ \mu_2(t)-F_2(t)x_{\mu}(t)\end{array}\right)^\top {\mathbb R}(t,X(t))\left(\begin{array}{c} \mu_1(t)-F_1(t)x_{\mu}(t)\\ \mu_2(t)-F_2(t)x_{\mu}(t)\end{array}\right) dt\right]
\end{align*}
for all $s\in(\tau, T)$, $\mu_k(\cdot)\in L_w^2([s,T], {\mathbb R}^{m_k}), k=1,2, x_s\in {\mathbb R}^n$, where ${\mathbb R}(t,X(t))$ is computed as in (\ref{e.106}) and $F(t)=(F_1(t)^\top \quad F_2(t)^\top)^\top $ is computed as in (\ref{e.22}).
Further, employing partition (\ref{e.15}) we may write:
 \begin{align}\label{3.1004}
&\mathring{\mathcal{J}}(s,T,x_s;\mu_1(\cdot),\mu_2(\cdot))=\frac{1}{2}\langle X(s)x_s,x_s\rangle +\notag\\
&+\frac{1}{2} {\mathbb E}\left[\int_s^T |{\mathbb V}_{21}(t)(\mu_1(t)-F_1(t)x_{\mu}(t))+{\mathbb V}_{22}(t)(\mu_2(t)-F_2(t)x_{\mu}(t))|^2 dt\right]-\notag\\
&-\frac{1}{2} {\mathbb E}\left[\int_s^T |{\mathbb V}_{11}(t)(\mu_1(t)-F_1(t)x_{\mu}(t))|^2 dt\right]
\end{align}
Let $\hat x(t), t\in[s,T]$ be the solution of the following IVP:
\begin{eqnarray}\label{3.1005}
dx(t)=[A_0(t)+B_{02}(t)K(t)+(B_{01}(t)+B_{02}(t)W(t))F_1(t)]x(t)dt+\nonumber\\
\sum_{j=1}^r [A_j(t)+B_{j2}(t)K(t)+(B_{j1}(t)+B_{j2}(t)W(t))F_1(t)]x(t)dw_j(t)\\
\hat x(s)= x_s\in \mathbb{R}^n\nonumber
\end{eqnarray}
We set  $\hat \mu_1(t):= F_1(t)\hat x(t)$,  $\hat \mu_2(t) :=K(t)\hat x(t)+W(t)\hat \mu_1(t)$ for $t\in[s,T]$.
With these notations, (\ref{3.1005}) may  be rewritten as:
\begin{eqnarray}\label{3.1006}
&d\hat x(t)=(A_0(t)\hat x(t)+B_{01}(t)\hat \mu_1(t)+B_{02}(t)\hat \mu_2(t))dt +\nonumber\\
&\sum_{j=1}^r (A_j(t)\hat x(t)+B_{j1}(t)\hat \mu_1(t)+B_{j2}(t)\hat \mu_2(t))\hat x(t)dw_j(t)\\
&x_{\hat \mu}(s)=x_s\in \mathbb{R}^n\nonumber
\end{eqnarray}
If $x_{\hat \mu}(t)$ is the solution of the IVP (\ref{3.1006}) we may infer that $x_{\hat \mu}(t)=\hat x(t), t\in[s,T]$.
Hence,
\begin{eqnarray}\label{3.1007}
&\hat \mu_1(t)= F_1(t)x_{\hat\mu}(t)\nonumber\\
&\hat \mu_2(t)= K(t)x_{\hat \mu}(t)+W(t)\hat \mu_1(t)
\end{eqnarray}
In this case, (\ref{3.1004}) yields
\begin{eqnarray}\label{3.1008}
\mathring{\mathcal{J}}(s,T,x_s;\hat{\mu}_1(\cdot),\hat{\mu}_2(\cdot))=\frac{1}{2}\langle X(s)x_s,x_s\rangle +\frac{1}{2} {\mathbb E}\left[\int_s^T |{\mathbb V}_{22}(t)(\hat \mu_2(t)-F_2(t)x_{\hat \mu}(t))|^2 dt\right]
\end{eqnarray}
On the other hand, (\ref{e1.4}), (\ref{e3.5}) and (\ref{3.1007}) allow us to obtain the equality:
\begin{eqnarray}\label{3.1009}
\mathring{\mathcal{J}}(s,T,x_s;\hat \mu_1(\cdot),\hat \mu_2(\cdot))=\mathring{\mathcal{J}}_{KW}(s,T,x_s;\hat \mu_1(\cdot))
\end{eqnarray}
for all $s\in(\tau, T]$ and $x_s\in{\mathbb R}^n$.
Finally, employing (\ref{3.1002}) together with (\ref{3.1008}) and (\ref{3.1009}) we obtain (\ref{3.1003}).
Thus the proof ends.
\end{proof}
\noindent Consider now $\left(\Phi(\cdot),\hat \Phi(\cdot)\right):[0,\; T]\rightarrow \mathbb{R}^{m_1\times n}\times \mathbb{R}^{m_1\times n}$ be a pair of continuous matrix valued functions. We set:
\begin{equation}\label{e.150}
u_1(t):=\Phi(t)x(t)+\left(\hat \Phi(t)-\Phi(t)\right)\mathbb{E}[x(t)]
\end{equation}
It follows that:
\begin{equation}\label{e.151}
\mathbb{E}[u_1(t)]=\hat \Phi(t)\mathbb{E}[x(t)]
\end{equation}
Substituting (\ref{e.150}) and (\ref{e.151}) in (\ref{e.5}) and (\ref{e.6}), we obtain:
\begin{align}\label{e.152}
&dx(t)=\left({\cal A}_{0\Phi}(t)x(t)+{\cal A}_{0\Phi\hat \Phi}(t)\mathbb{E}[x(t)]+ B_{02}(t)u_2(t)+\bar{B}_{02}(t)\mathbb{E}[u_2(t)]\right)dt+\notag\\
&+\sum_{j=1}^r\left({\cal A}_{j\Phi}(t)x(t)+{\cal A}_{j\Phi\hat \Phi}(t)\mathbb{E}[x(t)]+B_{j2}(t)u_2(t)+\bar{B}_{j2}(t)\mathbb{E}[u_2(t)]\right)dw_j(t)\notag\\
&x(s)=x_s\in \mathbb{R}^n
\end{align}
\begin{align}\label{e.153}
&\mathcal{J}_{\Phi}^{\hat \Phi}(s,T,x_s;u_2(\cdot))=\frac{1}{2}\mathbb{E}\left[\left \langle G_Tx_{u_2}(T),x_{u_2}(T)\right \rangle+\left \langle \hat{G}_T\mathbb{E}[x_{u_2}(T)],\mathbb{E}[x_{u_2}(T)]\right \rangle+\right.\notag\\
&\left.+\int_s^T\left(\left \langle \left(\begin{array}{cc}{\mathbb M}_\Phi(t) & {\mathbb L}_{\Phi}(t) \\\star & R_{22}(t)\end{array}\right)\left(\begin{array}{c}x_{u_2}(t) \\u_2(t)\end{array}\right),\left(\begin{array}{c}x_{u_2}(t) \\u_2(t)\end{array}\right) \right\rangle+\right.\right.\notag\\
&\left.\left.+\left \langle \left(\begin{array}{cc}{\mathbb M}_{\Phi\hat \Phi}(t) & {\mathbb L}_{\Phi\hat \Phi}(t) \\\star & \bar{R}_{22}(t)\end{array}\right)\left(\begin{array}{c}\mathbb{E}[x_{u_2}](t) \\\mathbb{E}[u_2(t)]\end{array}\right),\left(\begin{array}{c}\mathbb{E}[x_{u_2}(t)] \\\mathbb{E}[u_2(t)]\end{array}\right) \right\rangle\right)\right]
\end{align}
where $x_{u_2}(t)$ is the solution of the IVP (\ref{e.152}) determined by the input $u_2(\cdot)\in L_w^2\left([0,\; T];\mathbb{R}^{m_2}\right)$.\\
In (\ref{e.152}) we denoted:
\begin{subequations}\label{e.154}
\begin{align}
\label{e.154a}
&{\cal A}_{j\Phi\hat \Phi}(t)=\hat {\cal A}_{j\hat \Phi}(t)- {\cal A}_{j\Phi}(t)\\
\label{e.154b}
&\hat {\cal A}_{j\hat \Phi}(t)={\hat A}_j(t)+ {\hat B}_{j1}(t)\hat \Phi(t)\\
\label{e.154c}
& {\cal A}_{j\Phi}(t)={A}_j(t)+ {B}_{j1}(t) \Phi(t)
\end{align}
\end{subequations}
and in (\ref{e.153}) we have used the notations:
\begin{subequations}\label{e.155}
\begin{align}
\label{e.155a}
&{\mathbb M}_{\Phi\hat \Phi}(t)=\hat {\mathbb M}_{\hat \Phi}(t)- {\mathbb M}_{\Phi}(t)\\
\label{e.155b}
&\hat {\mathbb M}_{\hat \Phi}(t)=\left(\begin{array}{c}I_n \\\hat \Phi(t)\end{array}\right)^\top\left(\begin{array}{cc}\hat M(t) & {\hat L}_1(t) \\\star & {\hat R}_{11}(t)\end{array}\right)\left(\begin{array}{c}I_n \\\hat \Phi(t)\end{array}\right)\\
\label{e.155c}
& {\mathbb M}_{\Phi}(t)=\left(\begin{array}{c}I_n \\\Phi(t)\end{array}\right)^\top\left(\begin{array}{cc}M(t) & {L}_1(t) \\\star & {R}_{11}(t)\end{array}\right)\left(\begin{array}{c}I_n \\\Phi(t)\end{array}\right)
\end{align}
\end{subequations}
\begin{subequations}\label{e.156}
\begin{align}
\label{e.156a}
&{\mathbb L}_{\Phi \hat \Phi}(t)={\hat {\mathbb L}}_{\hat \Phi}(t)- {\mathbb L}_{\Phi}(t)\\
\label{e.156b}
&{\hat {\mathbb L}}_{\hat \Phi}(t)={\hat L}_2(t)-{\hat\Phi}^\top(t){\hat R}_{12}(t)\\
\label{e.156c}
& {\mathbb L}_{\Phi}(t)={L}_2(t)-{\Phi}^\top(t){R}_{12}(t)
\end{align}
\end{subequations}
The Riccati differential equations of type (\ref{e.9}), (\ref{e.10}), respectively, associated to the pair consisting of the controlled system (\ref{e.152}) and the quadratic functional (\ref{e.153}) are:
\begin{subequations}\label{e.157}
\begin{align}
\label{e.157a}
&-{\dot{\Upsilon}}(t)={\cal A}_{0\Phi}^{\top}(t)\Upsilon(t)+\Upsilon(t){\cal A}_{0\Phi}(t)+
\sum_{j=1}^r {\cal A}_{j\Phi}^{\top}(t)\Upsilon(t){\cal A}_{j\Phi}(t)-\notag\\
&-\Big(\Upsilon(t)B_{02}(t)+\sum_{j=1}^r {\cal A}_{j\Phi}^{\top}(t)\Upsilon(t)B_{j2}(t) +{\mathbb L}_{\Phi}(t)\Big)\Big(R_{22}(t)+\sum_{j=1}^r B_{j2}^{\top}(t)\Upsilon(t)B_{j2}(t)\Big)^{-1}\times\notag\\
&\times\Big(B_{0 2}^{\top}(t)\Upsilon(t)+\sum_{j=1}^r B_{j 2}^{\top}(t)\Upsilon(t){\cal A}_{j\Phi}(t)+ {\mathbb L}_{\Phi}^{\top}(t)\Big)+ {\mathbb M}_{\Phi}(t)\\
\label{e.157b}
&\Upsilon(T)=G_T
\end{align}
\end{subequations}
\begin{subequations}\label{e.158}
\begin{align}
\label{e.158a}
&-{\dot{\hat \Upsilon}}(t)={\hat {\cal A}}_{0\hat \Phi}^{\top}(t)\hat \Upsilon(t)+\hat \Upsilon(t){\hat {\cal A}}_{0\hat \Phi}(t)+
\sum_{j=1}^r{\hat {\cal A}}_{j\hat \Phi}^{\top}(t)\Upsilon(t){\hat {\cal A}}_{j\hat \Phi}(t)-\notag\\
&-\Big(\hat \Upsilon(t){\hat B}_{02}(t)+\sum_{j=1}^r{\hat {\cal A}}_{j\hat \Phi}^{\top}(t)\Upsilon(t){\hat B}_{j2}(t)+{\hat {\mathbb L}}_{\hat \Phi}(t)\Big)\Big({\hat R}_{22}(t)+\sum_{j=1}^r{\hat B}_{j2}^{\top}(t)\Upsilon(t){\hat B}_{j2}(t)\Big)^{-1}\times\notag\\
&\times\Big({\hat B}_{02}^{\top}(t)\hat \Upsilon(t)+\sum_{j=1}^r{\hat B}_{j2}^{\top}(t)\Upsilon(t){\hat {\cal A}}_{j\hat \Phi}(t)+{\hat {\mathbb L}}_{\hat \Phi}^{\top}(t)\Big)+{\hat {\mathbb M}}_{\hat \Phi}(t)\\
\label{e.146b}
&\hat \Upsilon(T)={\hat G}_T
\end{align}
\end{subequations}
Applying Proposition \ref{P3.2} in the case of the pair consisting of the controlled system (\ref{e.152}) and the quadratic functional (\ref{e.153}), we obtain:
\begin{cor}\label{C3.8}
For the controlled system (\ref{e.152}) and the quadratic functional (\ref{e.153}) associated to the pair of continuous matrix valued functions $\left(\Phi(\cdot),\hat \Phi(\cdot)\right)$, the following are equivalent:
\item [i)] There exists $\delta >0$ with the property that:
\begin{equation}\label{e.159}
\mathcal{J}_{\Phi}^{\hat \Phi}(0,T,0;u_2(\cdot))\geq  \delta \mathbb{E}\left[\int_0^T|u_2(t)|^2dt\right]
\end{equation}
$\forall u_2(\cdot)\in L_w^2\left([0,\; T];\mathbb{R}^{m_2}\right)$.
\item [ii)] The TVPs of type (\ref{e.157}), (\ref{e.158}) have the solutions $\left(\Upsilon_\Phi(\cdot),\Upsilon_\Phi^{\hat \Phi}(\cdot)\right):[0,\; T]\rightarrow \mathcal{S}_n\times \mathcal{S}_n$ with the additional property that $\Upsilon_{\Phi}(t)$ satisfies the constraints:
\begin{equation}\label{e.160}
R_{22}(t)+\sum_{j=1}^rB_{j2}^\top(t)\Upsilon_{\Phi}(t)B_{j2}(t)\geq \nu I_{m_2}
\end{equation}
\begin{equation}\label{e.161}
{\hat R}_{22}(t)+\sum_{j=1}^r{\hat B}_{j2}^\top(t)\Upsilon_{\Phi}(t){\hat B}_{j2}(t)\geq \nu I_{m_2}
\end{equation}
for all $t\in [0,\; T]$, where $\nu$, $\hat\nu$ are positive constants.
\end{cor}
\noindent Applying Lemma \ref{L3.1} in the case of the quadratic functional (\ref{e.153}) for which (\ref{e.159}) holds, we obtain via (\ref{e.160})- (\ref{e.161}) that
\begin{eqnarray}\label{3.1010}
\mathcal{J}_{\Phi}^{\hat \Phi}(s,T,x_s;u_2(\cdot)) \geq \frac{1}{2}
\left\langle \Upsilon_{\Phi}^{\hat \Phi}(s) x_s, x_s \right\rangle
\end{eqnarray}
for all $s\in [0,T)$, $u_2(\cdot)\in L_w^2([s,T], {\mathbb R}^{m_2})$, $x_s\in{\mathbb R}^n$.\\
\noindent Motivated by the result stated in Corollary \ref{C3.8}, we introduce the condition:\\
\noindent \textit{\textbf{C2}) There exist continuous matrix valued functions $\left(\Phi(\cdot),\hat \Phi(\cdot)\right):[0,\; T]\rightarrow \mathbb{R}^{m_1\times n}\times \mathbb{R}^{m_1\times n}$ with the property that the mapping $u_2(\cdot)\rightarrow \mathcal{J}_{\Phi}^{\hat \Phi}(0,T,0;u_2(\cdot)):L_w^2\left([0,\; T];\mathbb{R}^{m_2}\right)\rightarrow \mathbb{R}$ is uniformly convex, that is there exists $\delta >0$ for which (\ref{e.159}) holds.}\\
\noindent Setting formally $\mu_1(t) \equiv \mu_{1\Phi}(t)=\Phi(t)x(t)$ in ({\ref{e1.3}}) and ({\ref{e1.4}}), we obtain:
\begin{align}\label{e3.4bis}
dx(t)&=({\cal A}_{0\Phi}(t)x(t)+B_{02}(t)\mu_2(t))dt+\sum\limits_{j=1}^r({\cal A}_{j\Phi}(t)x(t)+B_{j2}(t)\mu_2(t))dw_j(t),\quad x(s)=x_s\in{\mathbb{R}}^n
\end{align}
\begin{align}\label{e3.5bis}
\mathring{\mathcal{J}}_{\Phi}(s,T,x_s;\mu_2(\cdot))=\frac{1}{2}\mathbb{E}\left[\left\langle G_Tx_{\mu_2}(T),x_{\mu_2}(T)\right\rangle+\int\limits_0^T\left\langle\left(
                                                \begin{array}{cc}
                                                  {\mathbb M}_{\Phi}(t) & {\mathbb L}_{\Phi}(t) \\
                                                  {\mathbb L}_{\Phi}^\top(t) & R_{22}(t) \\
                                                \end{array}
                                              \right)\left(
                                       \begin{array}{c}
                                         x_{\mu_2}(t) \\
                                         \mu_{2}(t) \\
                                       \end{array}
                                     \right),\left(
                                       \begin{array}{c}
                                         x_{\mu_2}(t) \\
                                         \mu_{2}(t) \\
                                       \end{array}
                                     \right)\right\rangle dt\right]
\end{align}
where $x_{\mu_2}(t)$ is the solution of the initial value problem ({\ref{e3.4bis}}) corresponding to the input $\mu_2(t)$.\\


\begin{lem}\label{L3.8}
a)  Assume that for the continuous matrix valued function $\Phi(\cdot): [0,T]\to {\mathbb R}^{m_1\times n}$, the solution $\Upsilon_{\Phi}(\cdot)$ of the TVP  (\ref{e.157}) is defined on the whole interval $[0,T]$ and satisfies the sign condition (\ref{e.160}).
Under these conditions we have:
\begin{eqnarray}\label{3.1011}
\mathring{\cal J}_{\Phi}(s, T, x_s; \mu_2(\cdot))\geq \frac{1}{2}\left\langle \Upsilon_{\Phi}(s)x_s, x_s\right\rangle
\end{eqnarray}
for all $s\in[0,T), \mu_2(\cdot)\in L_w^2([s,T]; {\mathbb R}^{m_2}), x_s\in{\mathbb R}^n$.\\

\noindent b) If the solution $X(\cdot)$ of the TVP (\ref{e.9}) is defined and satisfies the sign condition (\ref{e.111}) on an interval $(\tau, T]$  and if
$\Upsilon_{\Phi}(\cdot):[0,T]\to {\cal S}_n$ is the solution of the TVP (\ref{e.157})  satisfying the sign condition (\ref{e.160}) then
\begin{eqnarray}\label{3.1012}
X(s)\geq \Upsilon_{\Phi}(s)
\end{eqnarray}
for all $s\in(\tau, T]$.
\end{lem}

\begin{proof}
a)  Applying It\^o's formula in the case of the function $v_2(t,x)= x^\top \Upsilon_{\Phi}(t) x $ and to the stochastic process $x_{\mu_2}(t)$ defined as the solution of the IVP (\ref{e3.4bis})  we obtain:
\begin{align}\label{3.1013}
&\mathring {\cal J}_{\Phi}(s, T, x_s; \mu_2(\cdot))=\frac{1}{2} \left\langle \Upsilon_{\Phi}(s) x_s, x_s\right\rangle +\notag\\
&+\frac{1}{2} {\mathbb E}\left[\int_s^T \left\langle \left(R_{22}(t)+\sum_{j=1}^r B_{j2}^\top(t)\Upsilon_{\Phi}(t)B_{j2}(t)\right)\left(\mu_2(t)-F_{\Phi}(t)x_{\mu_2}(t)\right), \mu_2(t)-F_{\Phi}(t)x_{\mu_2}(t)\right\rangle dt\right]
\end{align}
where $F_{\Phi}(t):= -(R_{22}(t)+\sum_{j=1}^r B_{j2}^\top(t)\Upsilon_{\Phi}(t)B_{j2}(t))^{-1}(B_{02}^\top(t)\Upsilon_{\Phi}(t)+\sum_{j=1}^r B_{j2}^\top(t)\Upsilon_{\Phi}(t){\cal A}_{j\Phi}(t)+{\mathbb L}_{\Phi}^\top(t))$.
Now, (\ref{3.1011}) follows from (\ref{3.1013}) together with (\ref{e.160}).\\

\noindent b)  Under the considered assumptions the quadratic functional (\ref{e1.4})  takes the form given in (\ref{3.1004}) for all 
$s\in (\tau, T]$, $\mu_k(\cdot) \in L_w^2([s,T]; {\mathbb R}^{m_k}), k=1,2, x_s\in{\mathbb R}^n$.
Let $\bar x(t),t \in [s,T]$ be the solution of the following IVP:
\begin{eqnarray}\label{3.1014}
dx(t)=[A_0(t)+(B_{01}(t)+B_{02}(t)\tilde W(t))\Phi(t)+B_{02}(t)\tilde K(t)] x(t)dt+\nonumber\\
\sum_{j=1}^r[A_j(t)+(B_{j1}(t)+B_{j2}(t)\tilde W(t))\Phi(t)+B_{j2}(t)\tilde K(t)]x(t)dw_j(t)\\
\bar x(s)=x_s\in \mathbb{R}^n\nonumber
\end{eqnarray}
where 
\begin{eqnarray}\label{3.1014bis}
&\tilde K(t)={\mathbb V}(t)^{-1} ( {\mathbb V}_{21}(t) \qquad  {\mathbb V}_{22}(t)) \; F(t)\\
&\tilde W(t)= - {\mathbb V}_{22}^{-1}(t) {\mathbb V}_{21}(t) \nonumber
\end{eqnarray}
with $F(t)$ defined in (\ref{e.22}). 
We set $\bar \mu(t)=(\bar \mu_1^\top(t) \qquad \bar \mu_2^\top(t) )^\top$ 
\begin{eqnarray}\label{3.1015}
&\bar \mu_1(t):= \Phi(t)\bar x(t)\\
&\bar \mu_2(t):=\tilde K(t) \bar x(t)+\tilde W(t)\bar \mu_1(t)\nonumber
\end{eqnarray}
With these notations (\ref{3.1014}) becomes 
\begin{eqnarray}\label{3.1016}
dx(t)=(A_0(t)x(t)+B_{01}(t)\bar \mu_1(t)+B_{02}(t)\bar\mu_2(t))dt+\nonumber\\
\sum_{j=1}^r (A_j(t)x(t)+B_{j1}(t)\bar\mu_1(t)+B_{j2}(t)\bar \mu_2(t))dw_j(t)\\
\bar x(s)=x_s\in \mathbb{R}^n\nonumber
\end{eqnarray}
If $x_{\bar\mu}(t)$ is the solution of the IVP (\ref{3.1016}) then (\ref{3.1015}) becomes 
\begin{eqnarray}\label{3.1017}
&\bar \mu_1(t)= \Phi(t)x_{\bar \mu}(t)\\
&\bar \mu_2(t)=\tilde K(t) x_{\bar\mu}(t)+\tilde W(t)\bar \mu_1(t)\nonumber
\end{eqnarray}
In this case (\ref{3.1004}) yields 
\begin{eqnarray}\label{3.1018}
\mathring{\mathcal{J}}(s,T,x_s;\bar \mu_1(\cdot),\bar \mu_2(\cdot))=\frac{1}{2}\langle X(s)x_s,x_s\rangle-
\frac{1}{2} {\mathbb E}\left[\int_s^T |{\mathbb V}_{11}(t)(\Phi(t)-F_1(t)) x_{\bar\mu}^1(t)|^2dt\right]
\end{eqnarray}
On the other hand, from (\ref{e1.4}), (\ref{e3.5bis})  and (\ref{3.1017}) we deduce that 
\begin{eqnarray}\label{3.1019}
\mathring{\mathcal{J}}(s,T,x_s;\bar \mu_1(\cdot),\bar \mu_2(\cdot))=\mathring{\mathcal{J}}_{\Phi}(s,T,x_s;\bar \mu_2(\cdot))
\end{eqnarray}
Employing (\ref{3.1011}), (\ref{3.1018}) and (\ref{3.1019}) we may conclude that (\ref{3.1012}) holds. Thus the proof is complete.
\end{proof}
\noindent The next two Lemmas provide information about the behavior of the solution $\hat X(\cdot)$ of the TVP (\ref{e.10}).



\begin{lem}\label{L3.11}
Assume that the solution $X(\cdot)$ of the TVP (\ref{e.9}) is defined and satisfies the sign condition (\ref{e.111}) on the interval $(\tau,\; T]$. Let $\mathcal{I}_{\hat X}\subset [0,\; T]$ be the maximal interval where the solution $\hat X(\cdot)$ of the TVP (\ref{e.10}) is defined. If $\left(K(\cdot),\; W(\cdot)\right)$ and $\left(\hat K(\cdot),\; \hat W(\cdot)\right)$ are two pairs of continuous matrix values functions satisfying \textbf{C1)} then we have:
\begin{equation}\label{3.1020}
\hat X(s)\leq Y_{KW}^{\hat K \hat W}(s)
\end{equation}
for all $s\in \mathcal{I}_{\hat X} \bigcap (\tau,\; T]$, $Y_{KW}^{\hat K \hat W}(\cdot)$ being the solution of the TVP (\ref{e.146}).
\end{lem}
\begin{proof}
Let $s\in \mathcal{I}_{\hat X}\bigcap(\tau,\; T]$. Let $\left(\xi^1(\cdot),\; \xi^2(\cdot)\right)$ be the solution of the following IVP:
\begin{equation}\label{3.1021}
\begin{cases}
d\xi^1(t)=\left[A_0(t)+(B_{01}(t)+B_{02}(t)W(t))F_1(t)+B_{02}(t)K(t)\right]\xi^1(t)dt+\\
+\sum_{j=1}^r\{\left[A_j(t)+(B_{j1}(t)+B_{j2}(t)W(t))F_1(t)+B_{j2}(t)K(t)\right]\xi^1(t)+\\
+\left[\hat{A}_j(t)+(\hat{B}_{j1}(t)+\hat{B}_{j2}(t)\hat{W}(t))\hat{F}_1(t)+\hat{B}_{j2}(t)\hat{K}(t)\right]\xi^2(t)\}dw_j(t)\\
d\xi^2(t)=\left[\hat{A}_0(t)+(\hat{B}_{01}(t)+\hat{B}_{02}(t)\hat{W}(t))\hat{F}_1(t)+\hat{B}_{02}(t)\hat{K}(t)\right]\xi^2(t)dt\\
\xi^1(s)=0\\
\xi^2(s)=x_s\in \mathbb{R}^n
\end{cases}
\end{equation}
where $F_1(t)=(I_{m_1}\; 0)F(t)$, $\hat{F}_1(t)=(I_{m_1}\; 0)\hat{F}(t)$, $F(t)$ and $\hat F(t)$ being computed as in (\ref{e.22}) and (\ref{e.23}), respectively. For each $t\in [s,\; T]$ we set:
\begin{subequations}\label{3.1022}
\begin{align}
\label{3.1022a}
&\check{u}_1^1(t)\triangleq F_1(t)\xi^1(t)\\
\label{3.1022b}
&\check{u}_1^2(t)\triangleq \hat{F}_1(t)\xi^2(t)\\
\label{3.1022c}
&\check{u}_1(t)\triangleq \check{u}_1^1(t)+ \check{u}_1^2(t)
\end{align}
\end{subequations}
\begin{subequations}\label{3.1023}
\begin{align}
\label{3.1023a}
&\check{u}_2^1(t)\triangleq K(t)\xi^1(t)+W(t)\check{u}_1^1(t)\\
\label{3.1023b}
&\check{u}_2^2(t)\triangleq \hat{K}(t)\xi^2(t)+\hat W(t)\check{u}_1^2(t)\\
\label{3.1023c}
&\check{u}_2(t)\triangleq \check{u}_2^1(t)+ \check{u}_2^2(t)
\end{align}
\end{subequations}
If $\xi(t)\triangleq \xi^1(t)+\xi^2(t)$, then (\ref{3.1021})-(\ref{3.1023}) allow us to obtain:
\begin{equation}\label{3.1024}
\begin{cases}
d\xi(t)=\left[A_0(t)\xi(t)+\bar A_0(t)\xi^2(t)+B_{01}(t)\check{u}_1(t)+\bar B_{01}(t)\check{u}_1^2(t)+B_{02}(t)\check{u}_2(t)+\bar B_{02}(t)\check{u}_2^2(t)\right]dt+\\
+\sum_{j=1}^r\left[A_j(t)\xi(t)+\bar A_j(t)\xi^2(t)+B_{j1}(t)\check{u}_1(t)+\bar B_{j1}(t)\check{u}_1^2(t)+B_{j2}(t)\check{u}_2(t)+\bar B_{j2}(t)\check{u}_2^2(t)\right]dw_j(t)\\
\xi(s)=x_s\in \mathbb{R}^n
\end{cases}
\end{equation}
From the first equation of (\ref{3.1021}) we deduce that:
\begin{equation*}
\begin{cases}
d \mathbb{E}[\xi^1(t)]=[A_0(t)+(B_{01}(t)+B_{02}W(t))F_1(t)+B_{02}(t)K(t)]\mathbb{E}[\xi^1(t)] dt\\
\mathbb{E}[\xi^1(s)]=0
\end{cases}
\end{equation*}
Hence, $\mathbb{E}[\xi^1(t)]=0$, $\forall t\in [s,\; T]$. From the second equation of (\ref{3.1021}) we may infer that $\xi^2(t)=\mathbb{E}[\xi^2(t)]$, $\forall t\in [s,\; T]$ if $x_s\in \mathbb{R}^n$. Thus we have obtained that $\mathbb{E}[\xi(t)]=\xi^2(t)$. Further (\ref{3.1022}) and (\ref{3.1023}) allow us to deduce that $\mathbb{E}[\check{u}_k(t)]=\check{u}_k^2(t)$, $t\in [s,\; T]$, $k=1,2$. Substituting these equalities in (\ref{3.1024}) we may conclude that the solution $\xi(\cdot)$ of the IVP (\ref{3.1024}) coincides with the solution $x_{\check{u}}$ of the IVP (\ref{e.5}) corresponding to the input $\check{u}(t)=(\check{u}_1^\top(t)\quad\check{u}_2^\top(t))^\top$. One can easily check that $\check{u}_1(t)=F_1(t)x_{\check{u}}^1(t)$ and $\check{u}_2(t)=\hat F_1(t)x_{\check{u}^2}^2(t)$. Hence the equality (\ref{e5.19}) written for $u_k(\cdot)$ replaced by $\check u_k(\cdot)$ yields:    
\begin{align}\label{3.1025}
&\mathcal{J}(s,T,x_s;\check u_1(\cdot),\check u_2(\cdot))=\frac{1}{2}{\mathbb E}[\langle \hat{X}(s)x_s,x_s\rangle]\notag\\
&+\frac{1}{2}\mathbb{E}\left[\int_s^T\left\{|{\mathbb{V}}_{22}(t)(\check u^1_2(t)-F_2(t)x^1_{\check u}(t))|^2+|{\hat{\mathbb{V}}}_{22}(t)(\check u^2_2(t)-\hat{F}_2(t)x^2_{\check u^2}(t))|^2\right\}dt\right]
\end{align}
By direct calculation one obtains from (\ref{3.1022}) and (\ref{3.1023}) that $\check u_2(t)=K(t)x_{\check u}(t)+W(t)\check u_1(t)+(\hat K(t)-K(t))\mathbb E[x_{\check u}(t)]+(\hat W(t)-W(t))\mathbb E[{\check u}(t)]$, $t\in [s,\; T]$. Based on this equality one obtains that:
\begin{equation}\label{3.1026}
\mathcal{J}(s,T,x_s;\check u_1(\cdot),\check u_2(\cdot))=\mathcal{J}_{KW}^{\hat K \hat W}(s,T,x_s;\check u_1(\cdot))
\end{equation}
Employing (\ref{3.1001}), (\ref{3.1025}), (\ref{3.1026}) we deduce that $\left \langle\hat X(s)x_s,x_s\right\rangle\leq \left \langle\hat Y_{KW}^{\hat K \hat W}(s)x_s,x_s\right\rangle$. Hence, (\ref{3.1020}) holds because $x_s$ is arbitrary in $\mathbb{R}^n$. Thus the proof is completed.
\end{proof}
\begin{lem}\label{L3.13}
Assume that the solution $X(\cdot)$ of the TVP (\ref{e.9}) is defined and satisfies the sign condition (\ref{e.111}) on the interval $(\tau,\; T]$. Let $\mathcal{I}_{\hat X}\subset [0,\; T]$ be the maximal interval where the solution $\hat X(\cdot)$ of the TVP (\ref{e.10}) is defined. If $\left(\Phi(\cdot),\; \hat \Phi(\cdot)\right)$ are two continuous matrix values functions satisfying \textbf{C2)} then we have:
\begin{equation}\label{3.1027bis}
\hat X(s)\geq \Upsilon_{\Phi}^{\hat \Phi}(s)
\end{equation}
for all $s\in \mathcal{I}_{\hat X} \bigcap (\tau,\; T]$, $\Upsilon_{\Phi}^{\hat \Phi}(\cdot)$ being the solution of the TVP (\ref{e.158}).
\end{lem}
\begin{proof}
Let $s\in \mathcal{I}_{\hat X} \bigcap (\tau,\; T)$ be arbitrary but fixed. Let $\left(\zeta^1(\cdot),\; \zeta^2(\cdot)\right)$ be the solution of the following IVP:
\begin{equation}\label{3.1027}
\begin{cases}
d\zeta^1(t)=\left[A_0(t)+B_{02}(t)\tilde K(t)+(B_{01}(t)+B_{02}(t)\tilde W(t))\Phi(t)\right]\zeta^1(t)dt+\\
+\sum_{j=1}^r\{\left[A_j(t)+B_{j2}(t)\tilde K(t)+(B_{j1}(t)+B_{j2}(t)\tilde W(t))\Phi(t)\right]\zeta^1(t)+\\
+\left[\hat{A}_j(t)+\hat{B}_{j2}(t)\check{K}(t)+(\hat{B}_{j1}(t)+\hat{B}_{j2}(t)\check{W}(t))\hat{\Phi}(t)\right]\zeta^2(t)\}dw_j(t)\\
d\zeta^2(t)=\left[\hat{A}_0(t)+\hat{B}_{02}(t)\check{K}(t)+(\hat{B}_{01}(t)+\hat{B}_{02}(t)\check{W}(t))\hat{\Phi}(t)\right]\zeta^2(t)dt\\
\zeta^1(s)=0\\
\zeta^2(s)=x_s\in \mathbb{R}^n
\end{cases}
\end{equation} 
where $\tilde K(t)$, $\tilde W(t)$ were introduced in (\ref{3.1014bis}), while $\check{K}(t)\triangleq \hat{\mathbb{V}}_{22}^-1(t)(\hat{\mathbb{V}}_{21}(t)\quad \hat{\mathbb{V}}_{22}(t))\hat F(t)$ and $\check{W}(t)\triangleq -\hat{\mathbb{V}}_{22}^-1(t)\hat{\mathbb{V}}_{21}(t)$, $\forall t\in [s,\; T]$.\\ 
Now we set:
\begin{equation}\label{3.1028}
\begin{cases}
\hat u_1^1(t)\triangleq \Phi(t)\zeta^1(t)\\
\hat u_1^2(t)\triangleq \hat \Phi(t)\zeta^2(t)\\
\hat u_1(t)=\hat u_1^1(t)+\hat u_1^2(t)
\end{cases}
\end{equation}
\begin{equation}\label{3.1029}
\begin{cases}
\hat u_2^1(t)\triangleq \tilde K(t)(t)\zeta^1(t)+\tilde W(t)\hat u_1^1(t)\\
\hat u_2^2(t)\triangleq \check K(t)(t)\zeta^2(t)+\check W(t)\hat u_1^2(t)\\
\hat u_2(t)=\hat u_2^1(t)+\hat u_2^2(t)
\end{cases}
\end{equation}
\begin{equation}\label{3.1030}
\zeta(t)=\zeta^1(t)+\zeta^2(t)
\end{equation}
By direct calculation one obtains from (\ref{3.1027})--(\ref{3.1030}) that $t\rightarrow \zeta(t)$ is the solution of the following IVP:
\begin{equation}\label{3.1031}
\begin{cases}
d\zeta(t)=\left[A_0(t)\zeta(t)+\bar A_0(t)\zeta^2(t)+B_{01}(t)\hat{u}_1(t)+\bar B_{01}(t)\hat{u}_1^2(t)+B_{02}(t)\hat{u}_2(t)+\bar B_{02}(t)\hat{u}_2^2(t)\right]dt+\\
+\sum_{j=1}^r\left[A_j(t)\zeta(t)+\bar A_j(t)\zeta^2(t)+B_{j1}(t)\hat{u}_1(t)+\bar B_{j1}(t)\hat{u}_1^2(t)+B_{j2}(t)\hat{u}_2(t)+\bar B_{j2}(t)\hat{u}_2^2(t)\right]dw_j(t)\\
\zeta(s)=x_s\in \mathbb{R}^n
\end{cases}
\end{equation}
Proceeding as in the case of the IVP (\ref{3.1021}) one shows that in the case of the stochastic process defined in (\ref{3.1028})--(\ref{3.1030}) we have:
\begin{equation}\label{3.1032}
\begin{cases}
\mathbb{E}[\zeta(t)]=\zeta^2(t)\\
\mathbb{E}[\hat u_k(t)]=\hat u_k^2(t);\; k=1,2,\; t\in [s,\;T]
\end{cases}
\end{equation}
Substituting (\ref{3.1032}) in (\ref{3.1031}) we may conclude that the solution $\zeta(\cdot)$ of the IVP (\ref{3.1031}) coincides with the solution $x_{\hat u}(\cdot)$ of the IVP (\ref{e.5}) determined by the input $\hat u(t)=(\hat u_1^\top (t)\quad \hat u_2^\top (t))^\top$. Hence (\ref{3.1029}) yields:
\begin{equation*}
\mathbb{V}_{21}\left(\hat u_1^1(t)-F_1x_{\hat u}^1(t)\right)+\mathbb{V}_{22}\left(\hat u_2^1(t)-F_2x_{\hat u}^1(t)\right)=0\quad (\text{a.s})
\end{equation*}
\begin{equation*}
\hat{\mathbb{V}}_{21}\left(\hat u_1^2(t)-\hat F_1x_{\hat u^2}^2(t)\right)+\mathbb{V}_{22}\left(\hat u_2^2(t)-\hat F_2x_{\hat u^2}^2(t)\right)=0\quad (\text{a.s})
\end{equation*}
We recall that if $y(t)$ is a stochastic process, then $\hat y(t)$ stands for $y(t)-\mathbb{E}[y(t)]$ and $y^2(t)$ stands for $\mathbb{E}[y(t)]$. Hence, the equality (\ref{e5.19}) written for $u_k(\cdot)$ replaced by $\hat u_k(\cdot)$, $k=1,2$, becomes:
\begin{align}\label{3.1033}
&\mathcal{J}(s,T,x_s;\hat u_1(\cdot),\hat u_2(\cdot))=\frac{1}{2}{\mathbb E}[\langle \hat{X}(s)x_s,x_s\rangle]\notag\\
&-\frac{1}{2}\mathbb{E}\left[\int_s^T\left\{|{\mathbb{V}}_{11}(t)(\hat u^1_1(t)-F_1(t)x^1_{\hat u}(t))|^2+|{\hat{\mathbb{V}}}_{11}(t)(\hat u^2_1(t)-\hat{F}_1(t)x^2_{\hat u^2}(t))|^2\right\}dt\right]
\end{align}
On the other hand, (\ref{3.1028}) and (\ref{3.1032}) lead to $\hat u_1(t)=\Phi(t)x_{\hat u}(t)+\left(\hat \Phi(t)-\Phi(t)\right)\mathbb{E}[x_{\hat u}(t)]$. Substituting this equality in (\ref{e.6}) written for $u_k(\cdot)$ replaced by $\hat u_k(\cdot)$, $k=1,2$, we obtain:
\begin{equation}\label{3.1034}
\mathcal{J}(s,T,x_s;\hat u_1(\cdot),\hat u_2(\cdot))=\mathcal{J}_\Phi^{\hat \Phi}(s,T,x_s;\hat u_2(\cdot))
\end{equation}
for all $x_s\in \mathbb{R}^n$. Invoking (\ref{3.1010}) together with (\ref{3.1033}) and (\ref{3.1034}) we get: $\left\langle\Upsilon_\Phi^{\hat\Phi}(s)x(s),x(s)\right\rangle\leq\left\langle\hat X(s)x(s),x(s)\right\rangle$ for all $x_s\in \mathbb{R}^n$ which is equivalent to (\ref{3.1027bis}). Thus the proof ends.
\end{proof}
\noindent We are now in position to state and prove the main result of this section:
\begin{thm} The following are equivalent:
\begin{itemize}
\item [i)] the solutions $X(\cdot)$, $\hat X(\cdot)$ of the TVPs (\ref{e.9}), (\ref{e.10}), are defined on the whole interval $[0,\; T]$ and additionally $X(t)$ satisfies the conditions (\ref{e.111})-(\ref{e.112}).
\item [ii)] the conditions \textbf{C1}) and \textbf{C2}) are fulfilled;
\end{itemize}
\end{thm}
\begin{proof}
\begin{itemize}
\item []  i)$\Rightarrow$ii) If $X(\cdot),\; \hat X(\cdot)$ are well defined on the whole interval $[0;\; T]$, we may define  $F(t)$ and $\hat{F}(t)$ via (\ref{e.22}) and (\ref{e.23}), respectively. By setting:
\begin{equation}
\begin{cases}
\Phi(t)\triangleq \left(\begin{array}{cc}I_{m_1} &0\end{array}\right)F(t)\\
\hat{\Phi}(t)\triangleq \left(\begin{array}{cc}I_{m_1} &0\end{array}\right)\hat{F}(t)
\end{cases}
\end{equation}
one can shows, using Lemma 5.1.1 from \cite{carte2013} and relying on adequate algebraic manipulations, that the TVPs (\ref{e.9}) and (\ref{e.10}) can be equivalently rewritten as (\ref{e.157})-(\ref{e.158}), respectively. Because $X(\cdot)$ satisfies the sign conditions (\ref{e.111a}) and (\ref{e.112a}), it follows that the TVPs of type (\ref{e.157}), (\ref{e.158}) have the solutions $\left(\Upsilon_\Phi(\cdot),\Upsilon_\Phi^{\hat \Phi}(\cdot)\right):[0,\; T]\rightarrow \mathcal{S}_n\times \mathcal{S}_n$ with the additional property that $\Upsilon_{\Phi}(t)$ satisfies the constraints (\ref{e.160})-(\ref{e.161}). Hence, from Corollary 3.8, it follows that the condition \textbf{C2}) is fulfilled.\\
In a similar way, by using the factorization (\ref{e.15})-(\ref{e.18}) on the interval $[0, T]$ and setting:
\begin{equation}
\begin{cases}
K(t)\triangleq (\mathbb{V}_{22}(t))^{-1}\left(\begin{array}{cc}\mathbb{V}_{21}(t) & \mathbb{V}_{22}(t)\end{array}\right)F(t)\\
W(t)\triangleq-(\mathbb{V}_{22}(t))^{-1}\mathbb{V}_{21}(t)
\end{cases}
\end{equation}
and:
\begin{equation}
\begin{cases}
\hat{K}(t)\triangleq (\hat{\mathbb{V}}_{22}(t))^{-1}\left(\begin{array}{cc}\hat{\mathbb{V}}_{21}(t) & \hat{\mathbb{V}}_{22}(t)\end{array}\right)\hat{F}(t)\\
\hat{W}(t)\triangleq-(\hat{\mathbb{V}}_{22}(t))^{-1}\hat{\mathbb{V}}_{21}(t)
\end{cases}
\end{equation}
one can shows that the TVPs (\ref{e.9}) and (\ref{e.10}) can be equivalently rewritten as (\ref{e.145})-(\ref{e.146}), respectively. Because $X(\cdot)$ satisfies the sign conditions (\ref{e.111b}) and (\ref{e.112b}), it follows that the TVPs of type (\ref{e.145}), (\ref{e.146}) have the solutions $\left(Y_{KW}(\cdot),Y_{K W}^{\hat K \hat W}(\cdot)\right):[0,\; T]\rightarrow \mathcal{S}_n\times \mathcal{S}_n$ with the additional property that $Y_{KW}(t)$ satisfies the constraints (\ref{e.148})-(\ref{e.149}). Hence, from Corollary 3.5, it follows that the condition \textbf{C1}) is fulfilled.
\item []  ii)$\Rightarrow$i) Let $\Upsilon_\Phi(\cdot)$ be the solution of the TVP (\ref{e.157}) satisfying the sign conditions (\ref{e.160})-(\ref{e.161}), for $t\in [0;\; T]$. Hence, it follows that:
\begin{equation}\label{Ix.1}
R_{22}(T)+\sum_{j=1}^rB_{j2}^{\top}(T)G_TB_{j2}(T)>0
\end{equation}
\begin{equation}\label{Ix.1bis}
\hat{R}_{22}(T)+\sum_{j=1}^r\hat{B}_{j2}^{\top}(T)G_T\hat{B}_{j2}(T)>0
\end{equation}
Also, let $Y_{KW}(\cdot)$ be the solution of the TVP (\ref{e.145}) satisfying the sign condition (\ref{e.148})-(\ref{e.149}), for $t\in [0;\; T]$. Hence, it follows that:
\begin{equation}\label{Ix.2}
R_{W}(T)+\sum_{j=1}^rB_{jW}^{\top}(T)G_TB_{jW}(T)<0
\end{equation}
\begin{equation}\label{Ix.2bis}
\hat{R}_{W}(T)+\sum_{j=1}^r\hat{B}_{jW}^{\top}(T)G_T\hat{B}_{jW}(T)<0
\end{equation}
This yields that:
 \begin{align}\label{Ix.3}
&R_W(T)+\sum_{j=1}^rB_{jW}^{\top}(T)G_TB_{jW}(T)\notag\\
&-\left(R_{12}(T)+W^{\top}(T)R_{22}(T)+\sum_{j=1}^rB_{jW}^{\top}(T)G_TB_{j2}(T)\right)\notag\\
&\times\left(R_{22}(T)+\sum_{j=1}^rB_{j2}^{\top}(T)G_TB_{j2}(T)\right)^{-1}\star<0
\end{align}
The left hand side of (\ref{Ix.3}) is the Schur complement of the (2,2)-block of the matrix:
\begin{equation}\label{Ix.4}
\Theta(T)=\mathbb{W}^{\top}(T)\mathbb{R}(T,G_T)\mathbb{W}(T)
\end{equation}
where $\mathbb{W}(T)=\left(\begin{array}{cc}I_{m_1} & 0 \\W(T) & I_{m_2}\end{array}\right)$.
Using (\ref{Ix.4}) one can show by direct calculation that the left hand side of (\ref{Ix.3}) is the Schur complement of the (2,2)-block of the matrix $\mathbb{R}(T,G_T)$. Hence it follows that:
 \begin{align}\label{Ix.5}
&R_{11}(T)+\sum_{j=1}^rB_{j1}^{\top}(T)G_TB_{j1}(T)-\left(R_{12}(T)+\sum_{j=1}^rB_{j1}^{\top}(T)G_TB_{j2}(T)\right)\times\notag\\
&\times\left(R_{22}(T)+\sum_{j=1}^rB_{j2}^{\top}(T)G_TB_{j2}(T)\right)^{-1}\star<0
\end{align}
Similarly, one can show that:
 \begin{align}\label{Ix.5bis}
&\hat{R}_{11}(T)+\sum_{j=1}^r\hat{B}_{j1}^{\top}(T)G_T\hat{B}_{j1}(T)-\left(\hat{R}_{12}(T)+\sum_{j=1}^r\hat{B}_{j1}^{\top}(T)G_T\hat{B}_{j2}(T)\right)\times\notag\\
&\times\left(\hat{R}_{22}(T)+\sum_{j=1}^r\hat{B}_{j2}^{\top}(T)G_T\hat{B}_{j2}(T)\right)^{-1}\star<0
\end{align}
Using (\ref{Ix.1}) and (\ref{Ix.5}), it follows from the local existence theory of ODEs that $\mathcal{I}_X$ is not empty.\\
Now by applying the general theory of ODEs and based on the continuity of the coefficients of the Riccati equation (\ref{e.9}) one deduces that the solution $X(\cdot)$ of the Riccati equation (\ref{e.9}) is defined and satisfies the conditions (\ref{e.111})-(\ref{e.112}) on an interval $(\tau, T]\subseteq \mathcal{I}_X$. Let us assume that this is the maximal interval with these properties. Applying Lemma 3.6 and Lemma 3.9 one deduces that:
\begin{equation}\label{e.new1} 
\Upsilon_{\Phi}(s)\leq X(s)\leq Y_{KW}(s),\; s\in (\tau, T]  
\end{equation}
\noindent From the result above, it follows from one side that $X(\cdot)$ verifies the sign conditions (\ref{e.111a}) and (\ref{e.112a}) on $(\tau, T] $. On the other hand one has:  
 \begin{align}\label{equa4_main}
&R_W(t)+\sum_{j=1}^rB_{jW}^{\top}(t)X(t)B_{jW}(t)\notag\\
&-\left(R_{12}(t)+W^{\top}(t)R_{22}(t)+\sum_{j=1}^rB_{jW}^{\top}(t)X(t)B_{j2}(t)\right)\notag\\
&\times\left(R_{22}(t)+\sum_{j=1}^rB_{j2}^{\top}(t)X(t)B_{j2}(t)\right)^{-1}\star\llcurly0,\; t\in(\tau, T] 
\end{align}
 \begin{align}\label{equa4_mainbis}
&\hat{R}_{\hat{W}}(t)+\sum_{j=1}^r\hat{B}_{j\hat{W}}^{\top}(t)X(t)\hat{B}_{j\hat{W}}(t)\notag\\
&-\left(\hat{R}_{12}(t)+\hat{W}^{\top}(t)\hat{R}_{22}(t)+\sum_{j=1}^r\hat{B}_{j\hat{W}}^{\top}(t)X(t)\hat{B}_{j2}(t)\right)\notag\\
&\times\left(\hat{R}_{22}(t)+\sum_{j=1}^r\hat{B}_{j2}^{\top}(t)X(t)\hat{B}_{j2}(t)\right)^{-1}\star\llcurly0,\; t\in (\tau, T] 
\end{align}
The left hand side of (\ref{equa4_main}) is the Schur complements of the (2,2)-block of the matrix:
\begin{equation}\label{equa5_main}
\Theta(t)=\mathbb{W}^{\top}(t)\mathbb{R}(t,{X}(t))\mathbb{W}(t),\; t\in \mathcal{I}_\mathbf{X}
\end{equation}
where $\mathbb{W}(t)=\left(\begin{array}{cc}I_{m_1} & 0 \\W(t) & I_{m_2}\end{array}\right)$. Using (\ref{equa5_main}) one can show by direct calculation that the left hand side of (\ref{equa4_main}) is the Schur complement of the (2,2)-block of the matrix $\mathbb{R}(t,{X}(t))$. This yields that ${X}(t)$ verify the sign condition (\ref{e.111b}) on $(\tau, T]$. A similar reasoning as above yields to the conclusion that ${X}(t)$ verify the sign condition (\ref{e.112b}) on $(\tau, T]$.\\
Now by taking the limit for $t\rightarrow \tau$ one obtains that $X(\cdot) $ is define also at $t=\tau$ and  
satisfies also the sign conditions (\ref{e.111})-(\ref{e.112}).\\
Further, again from the general theory of ODEs and based on the continuity of the coefficients of the Riccati equation (\ref{e.9}) one deduces that the solution of the equation (\ref{e.9}) can be extended to  another interval $(\tau_1,\; \tau]$ and it satisfies the sign conditions(\ref{e.111})-(\ref{e.112}) on this interval. In this way, the maximality of the interval $(\tau, T]$ is  violated. Hence (\ref{e.new1}) remains true on $\mathcal{I}_X$. In addition, $X(t)$ being bounded on $\mathcal{I}_{X}$, it can be extended to $[0;\; T]$.\\
It remains now to prove that $\mathcal{I}_{\hat{X}}=[0;\; T]$. First, by using similar arguments as the ones used to prove that $\mathcal{I}_{X}\neq \emptyset$, we can show that $\mathcal{I}_{\hat{X}}$ is not empty. Hence, in order to get the desired conclusion it is sufficient to show that $\hat X(t)$ is uniformly bounded on $\mathcal{I}_{\hat{X}}$. Here also, one can apply a similar procedure as above in order to get the expected result. Indeed, thanks to Lemma 3.11 and Lemma 3.13 one can show that:
\begin{equation}
\Upsilon_{\Phi}^{\hat \Phi}(t)\leq \hat X(t)\leq Y_{KW}^{\hat K \hat W}(t);\quad t\in \mathcal{I}_{\hat{X}}
\end{equation}
This ends the proof.
\end{itemize}
\end{proof}
As already discussed in Section 2, in \cite{Sun:21bis} the authors provided sufficient conditions which guarantee the global existence of the whole interval $[0,\; T]$ of the solutions $X(\cdot)$, $\hat X(\cdot)$ of the RDEs (\ref{e.9}) and (\ref{e.10}) satisfying the sign conditions (\ref{e.109})-(\ref{e.110}). One can see evidently that if (\ref{e.109})-(\ref{e.110}) are satisfied then the sign conditions (\ref{e.111})-(\ref{e.112}) are also verified. The problem considered in \cite{Sun:21bis} can then be viewed as a particular case of the problem addressed in this paper. As a matter of fact, by specializing the result given in Theorem 3.15 one can obtain necessary and sufficient conditions which guarantee the global existence of the whole interval $[0,\; T]$ of the solutions $X(\cdot)$, $\hat X(\cdot)$ of the RDEs (\ref{e.9}) and (\ref{e.10}) satisfying the sign conditions (\ref{e.109})-(\ref{e.110}). To this end, let us first introduce the following condition:\\

\noindent \textit{\noindent \textbf{C}3) There exist continuous matrix valued functions $\left(K(\cdot),\hat K(\cdot)\right):[0,\; T]\rightarrow \mathbb{R}^{m_2\times n}\times \mathbb{R}^{m_2\times n}$ with the property that the mapping $u_1(\cdot)\rightarrow \mathcal{J}_{K}^{\hat K}(0,T,0;u_1(\cdot)):L_w^2\left\{[0,\; T];\mathbb{R}^{m_1}\right\}\rightarrow \mathbb{R}$ is uniformly concave, where $\mathcal{J}_{K}^{\hat K}(0,T,0;u_1(\cdot))$ is obtained from $\mathcal{J}_{KW}^{\hat K\hat W}(0,T,0;u_1(\cdot))$ by taking $W(t)=0$ and $\hat W(t)=0$, $t\in [0,\; T]$.}\\

\noindent We can now state the following result which improves the result given in Theorem 4.2 from \cite{Sun:21bis} where only sufficient conditions were provided. 

\begin{thm}The following are equivalent:
\begin{itemize}
\item [i)] the solutions $X(\cdot)$, $\hat X(\cdot)$ of the TVPs (\ref{e.9}), (\ref{e.10}), are defined on the whole interval $[0,\; T]$ and additionally $X(t)$ satisfies the conditions (\ref{e.109})-(\ref{e.110}).
\item [ii)] the conditions \textbf{C2}) and \textbf{C3}) are fulfilled;
\end{itemize}
\end{thm}
\subsection{Numerical example}
The following example is taken from \cite{Sun:SIAM} where the author illustrated that there proposed global existence conditions for the Riccati equations are only sufficient. Here we will illustrate the necessity of our conditions.\\   
\noindent Let us consider:
\begin{equation}\label{ex.1}
\begin{cases}
dx(t)=(u_1(t)+u_2(t))dt,\; t\ [0,\; 1]\\
x(0)=x\in \mathbb{R}
\end{cases}
\end{equation}
The performance criterion is:
\begin{equation}\label{ex.2}
J(x;u_1(\cdot),u_2(\cdot))=\mathbb{E}\left[-2x^2(1)+\int_0^1\left(u_1^2(t)-\frac{2}{3}u_2^2(t)\right)dt\right]
\end{equation}
The corresponding Riccati equation is:
\begin{equation}\label{ex.3}
\begin{cases}
\dot X(t)-X^2(t)+\frac{3}{2}X^2(t)=0\\
X(1)=-2
\end{cases}
\end{equation}
Its solution is $X(t)=\frac{2}{t-2}$ which is defined on the whole interval $[0,\; 1]$. In (\ref{ex.1}) and (\ref{ex.2}) the mean field terms are zero. Hence, the other Riccati equation is:
\begin{equation}\label{ex.3bis}
\begin{cases}
\dot {\hat X}(t)-\hat X(t)\left(\begin{array}{cc}1 & 1\end{array}\right)\tilde R^{-1}(t)\left(\begin{array}{c}1 \\1\end{array}\right)\hat X(t)=0\\
\hat X(1)=0
\end{cases}
\end{equation}
where $\tilde R(t)\triangleq \left(\begin{array}{cc}1 & 0 \\0 & -\frac{2}{3}\end{array}\right)$. Its solution is $\hat X(t)=0$, $\forall t\in [0,\; 1]$. We will show that the conditions from Theorem 3.17 are necessary conditions for the global existence of the solutions of the two Riccati equations (\ref{ex.3}) and (\ref{ex.3bis}). More precisely, we shall show that there exist two pairs of continuous functions $\left(\Phi(\cdot),\; \hat \Phi(\cdot)\right)$ and $\left(K(\cdot),\; \hat K(\cdot)\right)$ such that:
\begin{equation}\label{ex.4}
\mathcal J_\Phi^{\hat \Phi}(0;u_1(\cdot))\geq \delta \mathbb E\left[\int_0^1u_1^2(t)dt\right]
\end{equation}
for all $u_1(\cdot)$ from the set of admissible controls of player 1 and:
\begin{equation}\label{ex.5}
\mathcal J_K^{\hat K}(0;u_2(\cdot))\leq -\hat\delta \mathbb E\left[\int_0^1u_2^2(t)dt\right]
\end{equation} 
for all $u_2(\cdot)$ from the set of admissible controls of player 2, where $\delta>0$, $\hat \delta >0$ are constants.\\
We will not go through all the computational details, note however that by using Lemma 2.3 from \cite{Sun:16} we show that the following choice of pairs of continuous functions:
\begin{equation}
\begin{cases}
\Phi(t)=\frac{3}{2}X(t)\\
\hat \Phi(t)=0\\
K(t)=X(t)\\
\hat K(t)=0,\; t\in [0,\; 1]
\end{cases}
\end{equation}
where $X(t)$ is the solution of the Riccati equation (\ref{ex.3}), leads to the desired result.   
\section{Conclusion}
In this paper, we proposed a Riccati-type approach in order to solve an LQ mean-field game problem with a leader-follower structure for a class of SDEs with McKean--Vlasov type. We have obtained a state-feedback representation of the pairs of strategies which achieve an open-loop Stackelberg equilibrium. Our solution relies on the solvability of a coupled Riccati-type equations with indefinite sign of there quadratic terms. In the second part of this paper, we have then obtained necessary and sufficient conditions for the existence of solutions of the involved coupled generalized Riccati equations verifying specific sign conditions. Our ongoing effort are devoted to the infinite horizon counterpart of the problem treated in this paper. We believe that the stabilizing solutions of adequately defined generalized Riccati equations will play a key role in the solution process. One of the main challenges here is the definition of an adequate stability concept.
\section{Appendix}

\begin{lem}
Assume that the Riccati equations (\ref{e.9}) and (\ref{e.10}) are solvable on $\mathcal{I}_{{X}}$ and $\mathcal{I}_{\hat{{X}}}$, respectively. Then:
\begin{itemize}
\item [i)]  for every $s\in \mathcal{I}_{{X}}\bigcap \mathcal{I}_{\hat{{X}}}$, $u=\{u(t)\}_{t\geq s}\in L_{w}^2([s,T];{\mathbb{R}}^m)$ and $x_s\in \mathbb{R}^n$ we have:
\begin{eqnarray}\label{e.620}
\mathcal{J}(s,T,x_s,u_1,u_2)=\frac{1}{2}\langle \hat{X}(s)x_s, x_s\rangle+\mathcal{J}(s,T,0,u_1-\tilde{u}_1,u_2-\tilde{u}_2)
\end{eqnarray}
\item[ii)] if for every $s\in \mathcal{I}_{{X}}\bigcap \mathcal{I}_{\hat{{X}}}$ and $u_2=\{u_2(t)\}_{t\geq s}\in L_{w}^2([s,T];{\mathbb{R}}^{m_2})$, the condition:
\begin{equation}\label{Assump_cvx1}
\mathcal{J}(s,T,0,0,u_2-\tilde{u}_2)\geq 0
\end{equation}
is fulfilled then:
\begin{eqnarray}\label{e.621}
\frac{1}{2}\langle \hat X(s)x_s, x_s\rangle\leq {\cal J}(s,T,x_s;\tilde u_1(\cdot), u_2(\cdot))
\end{eqnarray}
for all $u_2(\cdot)\in L_w^2([s,T]; {\mathbb R}^{m_2})$ and $x_s\in {\mathbb R}^n$.
\item[iii)] for every $s\in \mathcal{I}_{{X}}\bigcap \mathcal{I}_{\hat{{X}}}$ and $u_1=\{u_1(t)\}_{t\geq s}\in L_{w}^2([s,T];{\mathbb{R}}^{m_1})$ the condition:
\begin{equation}\label{Assump_ccv1}
\mathcal{J}(s,T,0,u_1-\tilde{u}_1,0)\leq 0
\end{equation}
is fulfilled then:
\begin{eqnarray}\label{e.622}
\frac{1}{2}\langle \hat X(s)x_s, x_s\rangle\geq {\cal J}(s,T,x_s; u_1(\cdot), \tilde u_2(\cdot))
\end{eqnarray}
for all $u_1(\cdot)\in L_w^2([s,T]; {\mathbb R}^{m_1})$ and $x_s\in {\mathbb R}^n$.
\end{itemize}
\end{lem}
\begin{proof}
(i)  \; First note that $\mathcal{J}(s,T,x_s,u_1,u_2)=\mathcal{J}^1(s,T,0,u_1^1,u_2^1)+\mathcal{J}^2(s,T,x_s,u_1^2,u_2^2)$ where:
\begin{align}\label{J1}
\mathcal{J}^1(s,T,0,u_1^1,u_2^1)=\frac{1}{2}\mathbb{E}\left\{\left\langle G_Tx_u^1(T), x_u^1(T)\right\rangle+\int_s^T\left(\begin{array}{c}x_u^1(t) \\u_1^1(t) \\u_2^1(t)\end{array}\right)^\top\left(\begin{array}{ccc}M(t) & L_1(t) & L_2(t) \\\star & R_{11}(t) & R_{12}(t) \\\star & \star & R_{22}(t)\end{array}\right)\star dt\right\}
\end{align}
and
\begin{align}\label{J2}
\mathcal{J}^2(s,T,x_s,u_1^2,u_2^2)&=\frac{1}{2}\left\{\left\langle \hat{G}_Tx_u^2(T), x_u^2(T)\right\rangle+\int_s^T\left(\begin{array}{c}x_u^2(t) \\u_1^2(t) \\u_2^2(t)\end{array}\right)^\top\left(\begin{array}{ccc}\hat{M}(t) & \hat{L}_1(t) & \hat{L}_2(t) \\\star & \hat{R}_{11}(t) & \hat{R}_{12}(t) \\\star & \star & \hat{R}_{22}(t)\end{array}\right)\star dt\right\}
\end{align}
Let us first consider $\mathcal{J}^1(s,T,0,u_1^1,u_2^1)$. Since $x_u^1(t)=\tilde{x}_u^1(t)+\hat{x}_u^1(t)$, one obtains by direct computation that:
\begin{align}
&\mathcal{J}^1(s,T,0,u_1^1,u_2^1)=\mathcal{J}^1(s,T,0,\tilde{u}_1^1,\tilde{u}_2^1)+\mathcal{J}^1(s,T,0,(u_1^1-\tilde{u}_1^1),(u_2^1-\tilde{u}_2^1))+\mathbb{E}\left\{\left\langle G_T\hat{x}^1(T), \tilde{x}^1(T)\right\rangle\right\}\notag\\
&+\mathbb{E}\int_s^T\left\{\left\langle M(t)\hat{x}^1(t),\tilde{x}^1(t)\right\rangle+\left\langle L(t)\hat{x}^1(t),\tilde{u}^1(t)\right\rangle+\left\langle L(t)\tilde{x}^1(t),(u^1(t)-\tilde{u}^1(t))\right\rangle\right.\notag\\
&\left.+\left\langle R(t)\tilde{u}^1(t),(u^1(t)-\tilde{u}^1(t))\right\rangle\right\}dt
\end{align}
In a similar way, one can show that:
\begin{align}
&\mathcal{J}^2(s,T,x_s,u_1^2,u_2^2)=\mathcal{J}^2(s,T,x_s,\tilde{u}_1^2,\tilde{u}_2^2)+\mathcal{J}^2(s,T,0,(u_1^2-\tilde{u}_1^2),(u_2^2-\tilde{u}_2^2))+\left\langle \hat{G}_T\hat{x}^2(T), \tilde{x}^2(T)\right\rangle\notag\\
&+\int_s^T\left\{\left\langle \hat{M}(t)\hat{x}^2(t),\tilde{x}^2(t)\right\rangle+\left\langle \hat{L}(t)\hat{x}^2(t),\tilde{u}^2(t)\right\rangle+\left\langle \hat{L}(t)\tilde{x}^2(t),(u^2(t)-\tilde{u}^2(t))\right\rangle\right.\notag\\
&\left.+\left\langle \hat{R}(t)\tilde{u}^2(t),(u^2(t)-\tilde{u}^2(t))\right\rangle\right\}dt
\end{align}
Using the It\^o formula to systems (\ref{e.28}), (\ref{e.29}) and to the function\\ $\tilde{v}(t,\tilde{x}^1,\tilde{x}^2)=\frac{1}{2}\left(\left\langle X(t)\tilde{x}^1,\tilde{x}^1\right\rangle+\left\langle \hat{X}(t)\tilde{x}^2,\tilde{x}^2\right\rangle\right)$, $t\in [s,\; T]$, $\tilde{x}^1\in \mathbb{R}^n$, $\tilde{x}^2\in \mathbb{R}^n$ and taking into account the Riccati equations (\ref{e.9}) and (\ref{e.10}) we get:
\begin{equation}
\mathcal{J}^1(s,T,0,\tilde{u}_1^1,\tilde{u}_2^1)+\mathcal{J}^2(s,T,x_s,\tilde{u}_1^2,\tilde{u}_2^2)=\frac{1}{2}\left\langle \hat{X}(s) x_s, x_s\right\rangle
\end{equation}
Let:
\begin{align}
\Gamma_T^1&=\mathbb{E}\int_s^T\left\{\left\langle M(t)\hat{x}^1(t),\tilde{x}^1(t)\right\rangle+\left\langle L(t)\hat{x}^1(t),\tilde{u}^1(t)\right\rangle+\left\langle L(t)\tilde{x}^1(t),(u^1(t)-\tilde{u}^1(t))\right\rangle\right.\notag\\
&\left.+\left\langle R(t)\tilde{u}^1(t),(u^1(t)-\tilde{u}^1(t))\right\rangle\right\}dt
\end{align}
\begin{align}
\Gamma_T^2&=\int_s^T\left\{\left\langle \hat{M}(t)\hat{x}^2(t),\tilde{x}^2(t)\right\rangle+\left\langle \hat{L}(t)\hat{x}^2(t),\tilde{u}^2(t)\right\rangle+\left\langle \hat{L}(t)\tilde{x}^2(t),(u^2(t)-\tilde{u}^2(t))\right\rangle\right.\notag\\
&\left.+\left\langle \hat{R}(t)\tilde{u}^2(t),(u^2(t)-\tilde{u}^2(t))\right\rangle\right\}dt
\end{align}
Using the It\^o formula to systems (\ref{e.28})-(\ref{e.29}), (\ref{e.30})-(\ref{e.31}) and to the function\\ $\tilde{v}(t,\tilde{x}^1,\hat{x}^1,\tilde{x}^2,\hat{x}^2)=\left\langle X(t)\tilde{x}^1,\hat{x}^1\right\rangle+\left\langle \hat{X}(t)\tilde{x}^2,\hat{x}^2\right\rangle$, $t\in [s,\; T]$, $\tilde{x}^1\in \mathbb{R}^n$, $\hat{x}^1\in \mathbb{R}^n$, $\tilde{x}^2\in \mathbb{R}^n$, $\hat{x}^2\in \mathbb{R}^n$ and taking into account the Riccati equations (\ref{e.9}) and (\ref{e.10}) we get:
\begin{equation}
\Gamma_T^1+\Gamma_T^2=-\mathbb{E}\left\{\left\langle G_T\hat{x}^1(T), \tilde{x}^1(T)\right\rangle\right\}-\left\{\left\langle \hat{G}_T\hat{x}^2(T), \tilde{x}^2(T)\right\rangle\right\}
\end{equation}
This ends the proof of this implication.\\
\noindent (ii) \; The inequality (\ref{e.621}) is obtained from (\ref{e.620}) written for $u_1(\cdot)$ replaced by $\tilde u_1(\cdot)$ and taking into account the condition (\ref{Assump_cvx1}).\\
\noindent (iii) \;The inequality (\ref{e.622}) is obtained from (\ref{e.620}) written for $u_2(\cdot)$ replaced by $\tilde u_2(\cdot)$ and taking into account the condition (\ref{Assump_ccv1}). This ends the proof.
\end{proof}


\begin{thebibliography}{99}
 
 
\bibitem{Barreiro} J. Barreiro-Gomez, T. E. Duncan, and H. Tembine. Linear-Quadratic Mean-Field-Type Games: Jump-Diffusion Process With Regime Switching, \textit{IEEE Trans. Automat. Contr.}, 64, pp. 4329?4336, 2019.
 
 
\bibitem{Bensoussan} A. Bensoussan, J. Frehse, P. Yam. Mean Field Games and Mean
Field Type Control Theory. \textit{Springer}, New York, 2013.
%
\bibitem{Bensoussan1} A. Bensoussan, K. C. J. Sung, S. C. P. Yam, S. P. Yung. Linear quadratic
mean-field games. \textit{J. Optim. Theory Appl.}, 169:496--529, 2016.
%
\bibitem{Caines} P. E. Caines, M. Huang, R. P. Malham\'e. Mean Field Games. In \textit{Handbook of Dynamic Game Theory, T. Bassar and G. Zaccour Eds., Springer, Berlin}, 345-372, 2017.

\bibitem{Carmona} R. Carmona and F. Delarue. Probabilistic Theory of Mean Field Games with Applications, I?II, \textit{Springer Nature}, 2018.


%
\bibitem{carte2013} V. Dragan, T. Morozan and A.M. Stoica. Mathematical Methods in Robust Control of Linear Stochastic Systems, \textit{2nd edn, New York: Springer}, 2013.
%
%
\bibitem{OCAM}  V.~Dragan, S. Aberkane, T. Morozan. On the bounded and stabilizing solution of a generalized Riccati differential equation arising in connection with a zero-sum linear quadratic stochastic differential game. {\em Optimal Control Applications and Methods}, 41, 2020.
%
\bibitem{Gueant} O. Gu\'eant, J.-M. Lasry, P.-L. Lions. Mean field games and
applications. In \textit{Paris-Princeton Lectures on Mathematical Finance, Springer}, 205-266, 2011.

\bibitem{Graber} P. J. Graber. Linear quadratic mean field type control and mean field games with common noise, with application to production of an exhaustible resource, \textit{Appl. Math. Optim.}, 74, pp. 459?486, 2016.
%

\bibitem{Huang:06}
M. Huang, R. P. Malhame and P. E. Caines, Large population stochastic dynamic games: closed-loop McKean-Vlasov systems and the Nash certainty equivalence principle. \emph{Commun. Inf. Syst.}, 6(3), pp. 221--251, 2006.
\bibitem{Huang2} J. Huang, S. Wang, Z. Wu. Backward mean-field linear-quadratic-
Gaussian (LQG) games: Full and partial information. \textit{IEEE Trans. Autom. Control}, 61:3784--3796, 2016.
%
\bibitem{Kac} M. Kac. Foundations of kinetic theory, in: \textit{Proceedings of the Third Berkeley
Symposium on Mathematical Statistics and Probability}, III:171--197, 1956.

\bibitem{Lasry:06}
J. M. Lasry and P. L. Lions, Jeux a champ moyen. I. Le cas stationnaire. \emph{C. R. Math. Acad. Sci. Paris}, 343(9), pp. 619--625, 2006.
\bibitem{Lasry:06-1}
J. M. Lasry and P. L. Lions, Jeux a champ moyen. II. Horizon fini et contr\^ole optimal. \emph{C. R. Math. Acad. Sci. Paris}, 343(10), pp. 679--684, 2006.

\bibitem{Lasry:06-2}
J. M. Lasry and P. L. Lions, Mean field games.\textit{ Jpn. J. Math.}, 2(1), pp. 229--260, 2007.

\bibitem{Li} X. Li, J. Shi, and J. Yong. Mean-Field Linear-Quadratic Stochastic Differential Games in an
Infinite Horizon, \textit{arXiv}:2007.06130, 2020.
%
%
%
%
%
%
\bibitem{McKean} H.P. McKean. Propagation of chaos for a class of non-linear parabolic
equations, \textit{Lect. Ser. Differ. Equ}. 7:41--57, 1967.

\bibitem{Moon:20}
J. Moon, Linear-quadratic mean field stochastic zero-sum differential games. \textit{Automatica}, 120, 2020.

%

\bibitem{Stackelberg:34}
H. von Stackelberg. Marktform und Gleichgewicht, \textit{Springer}, Vienna, 1934.


\bibitem{Sun:16}
J. Sun, X. Li and J. Yong. Open-Loop and Closed-Loop Solvabilities For Stochastic Linear Quadratic Optimal Control Problems. \emph{SIAM Journal on Control and Optimization}, 54(5), 2274--2308, 2016.

\bibitem{Sun:17}
J. Sun. Mean-Field Stochastic Linear Quadratic Optimal Control Problems: Open-Loop Solvabilities. \emph{ESAIM: COCV}, 23, 1099--1127, 2017.


\bibitem{Sun:21}
J. Sun, H. Wang and J. Wen. Zero-Sum Stackelberg Stochastic Linear-Quadratic Differential Games. \emph{arXiv:2109.14893}, 2021.

\bibitem{Sun:21bis}
J. Sun, H. Wang and Z. Wu. Mean-Field Linear-Quadratic Stochastic Differential Games. \emph{Journal of Differential Equations}, 2021.

\bibitem{Sun:SIAM}
J. Sun. Two-Person Zero-Sum Stochastic Linear-Quadratic Differential Games. \emph{SIAM Journal on Control and Optimization}, 59(3), 1804--1829, 2021.

%
\bibitem{Tian} R. Tian, Z. Yu, R. Zhang. A closed-loop saddle point for zero-sum linear-quadratic stochastic differential games with mean-field type. \textit{Systems \& Control Letters}, 136, 2020.
%

\bibitem{Yong:02} J. Yong, A Leader-Follower Stochastic Linear Quadratic Differential Game, {\em SIAM Journal on Control and Optimization}, 41(4), 1015--1041, 2002.

\bibitem{Yong:13} J. Yong. Linear-Quadratic Optimal Control Problems for Mean-Field Stochastic Differential Equations, \textit{SIAM J. Control Optim.}, 51, pp. 2809--2838, 2013.

\bibitem{yong-2017} J. Yong, Linear-quadratic optimal control problems for mean-field stochastic differential equations -time consistent solutions, {\em Trans. Amer. Math. Soc.}, 369(8), 5467--5523, 2017.

\end{thebibliography}
\end{document}